\documentclass[UKenglish]{lmcs}
\pdfoutput=1

\usepackage{lastpage}
\renewcommand{\dOi}{14(4:6)2018}
\lmcsheading{}{1--31}{}{}%
{Feb.~02,~2016}{Oct.~29,~2018}{}

\pdfoutput=1

\usepackage{graphicx}
\usepackage{proof}
\usepackage[all]{xy}
\usepackage{amsmath}
\usepackage{amssymb}

\newcommand{\Bf}[1]{{\bf #1}}

\newcommand{\Rm}[1]{{\rm #1}}
\newcommand{\ul}{\underline}
\newcommand{\ol}{\overline}
\newcommand{\mc}{\mathcal}
\renewcommand{\angle}[1]{\langle #1\rangle}

\newcommand{\fa}[1]{\forall{#1}~.~}
\newcommand{\ex}[1]{\exists{#1}~.~}

\newcommand{\lam}[1]{\lambda{#1}~.~}

\newcommand{\Met}{\Bf{EPMet}} 

\newcommand{\Alg}{\Bf{Alg}}

\newcommand{\Obj}[1]{\Bf{Obj}(#1)}

\newcommand{\Set}{\Bf{Set}}
\newcommand{\Pred}{\Bf{Pred}}

\newcommand{\CAT}{\Bf{CAT}}

\newcommand{\dArrow}{\mathbin{\dot\Arrow}}

\newcommand{\dtimes}{\mathbin{\dot\times}}

\newcommand{\arrow}{\rightarrow}

\newcommand{\Arrow}{\Rightarrow}
\newcommand{\bul}{\bullet}
\newcommand{\Id}{\Rm{Id}}
\newcommand{\id}{\Rm{id}}
\newcommand{\union}{\Rm{union}}

\renewcommand{\AA}{\mathbb A}
\newcommand{\BB}{\mathbb B}
\newcommand{\CC}{\mathbb C}
\newcommand{\NN}{\mathbb N}
\newcommand{\DD}{\mathbb D}
\newcommand{\EE}{\mathbb E}

\newcommand{\mT}{\mc T}
\newcommand{\mD}{\mc D}

\newcommand{\mG}{\mc G}

\newcommand{\Top}{\Bf{Top}}
\newcommand{\Meas}{\Bf{Meas}}
\newcommand{\Ran}[1]{\Bf{Ran}_{#1}}

\newcommand{\Lan}[1]{\Bf{Lan}_{#1}}

\newcommand{\Pre}{\Bf{Pre}}

\newcommand{\Up}{\Bf{Up}}

\newcommand{\Lift}{\Bf{Lift}}

\newcommand{\Cls}{\Bf{Cls}}
\newcommand{\Split}{\Bf{Split}}

\newcommand{\toptop}{{\top\top}}
\newcommand{\TT}[2][-]{\ifx#1- #2^{\top\!\!\top} \else #2^{\top\!\!\top(#1)} \fi}
\newcommand{\BRel}{\Bf{BRel}}
\newcommand{\ERel}{\Bf{ERel}}
\newcommand{\BX}{X}
\newcommand{\BY}{Y}
\newcommand{\open}[1]{\mc O_{#1}}
\newcommand{\meas}[1]{\Sigma_{#1}}
\newcommand{\liftparam}[6]{
  \xymatrix@=1pc{#2_{#3} & #4 \ar[r]^-{#6} \ar[l]_-{#5} & #1}
}
\newcommand{\liftparamem}[6]{
  \xymatrix@=1pc{#2^{#3} & #4 \ar[r]^-{#6} \ar[l]_-{#5} & #1}
}

\newcommand{\pbmark}[1][dr]{\save*!/#1-1.2pc/#1:(-1,1)@^{|-}\restore}

\newcommand{\Act}{\mathrm{Act}}

\theoremstyle{plain}
\newtheorem{proposition}[thm]{Proposition}
\newtheorem{theorem}[thm]{Theorem}
\newtheorem{lemma}[thm]{Lemma}
\newtheorem{corollary}[thm]{Corollary}
\theoremstyle{definition}
\newtheorem{example}[thm]{Example}
\newtheorem{definition}[thm]{Definition}

\begin{document}

\title[Codensity Lifting of Monads and Its Dual]{Codensity Lifting of Monads and Its Dual}

\author[S. Katsumata]{Shin-ya Katsumata\rsuper{a}}
\address{\lsuper{a}Kyoto University, Kitashirakawaoiwakecho, Sakyoku, Kyoto, 606-8502, Japan}
\email{\{satoutet,sinya\}@kurims.kyoto-u.ac.jp}
\thanks{Katsumata's current affiliation is National Institute of Informatics, 2-1-2 Hitotsubashi, Chiyodaku, Tokyo, 100-0003, Japan.}

\author[T. Sato]{Tetsuya Sato\rsuper{a}}
\thanks{Sato's current affiliation is University at Buffalo, the State Universiy of New York, 338Y Davis Hall, Buffalo, NY 14260-2500, USA}

\author[T. Uustalu]{Tarmo Uustalu\rsuper{b}}
\address{\lsuper{b}Dept.\ of Software Science, Tallinn University of Technology, 
Akadeemia tee 21B, 12618 Tallinn, Estonia}
\email{tarmo@cs.ioc.ee}
\thanks{Uustalu is now also with School of Computer Science, 
Reykjav\'ik University, Menntavegi 1, 101 Reykjav\'ik, Iceland.}



\begin{abstract}
  We introduce a method to lift monads on the base category of a
  fibration to its total category.  This method, which we call {\em
    codensity lifting}, is applicable to various fibrations which were
  not supported by its precursor, categorical $\top\top$-lifting.
  After introducing the codensity lifting, we illustrate some examples
  of codensity liftings of monads along the fibrations from the
  category of preorders, topological spaces and extended pseudometric
  spaces to the category of sets, and also the fibration from the
  category of binary relations between measurable spaces. 
  We also introduce the dual method called {\em density lifting} of comonads. We next
  study the liftings of {\em algebraic operations} to the
  codensity liftings of monads. We also give a characterisation of the
  class of liftings of monads along posetal fibrations with fibred small
  meets as a limit of a certain large diagram.
\end{abstract}

\maketitle


\section{Introduction}
\label{sec:}

Inspired by Lindley and Stark's work on extending the concept of
reducibility candidates to monadic types \cite{lindley,lindleystark},
the first author previously introduced its semantic analogue called
{\em categorical $\top\top$-lifting} in \cite{katsumatatt}. It
constructs a lifting of a strong monad $\mT$ on the base category of a
closed-structure preserving fibration $p:\EE\arrow\BB$  to its total category.
The construction takes the inverse image of the continuation monad on
the total category along the canonical monad morphism
$\sigma:\mT\arrow(-\Arrow TR)\Arrow TR$ in the base category, which exists
for any strong monad $\mc T$:
\begin{equation}
  \label{eq:ttlift}
  \xymatrix{
    \TT \mT \ar@{.>}[rr] & & (- \Arrow S)\Arrow S \\
    \mT \ar[rr]_-\sigma & & (- \Arrow TR) \Arrow TR
  }
\end{equation}
The objects $R$ and $S$ (such that $TR=pS$) are presupposed parameters
of this $\top\top$-lifting, and by varying them we can derive various
liftings of $\mT$. The categorical $\top\top$-lifting has been used to
construct logical relations for monads \cite{katsumatarelating}
and to analyse the concept of preorders on monads \cite{DBLP:conf/fossacs/KatsumataS13}.

One key assumption for the $\top\top$-lifting to work is that the
fibration $p$ preserves the {\em closed structure}, so that the
continuation monad $(-\Arrow S)\Arrow S$ on the total category becomes
a lifting of the continuation monad $(-\Arrow TR)\Arrow TR$ on the base category. Although many such
fibrations are seen in the categorical formulations of logical
relations \cite{mitsce,hermidathesis,katsumatarelating},
requiring fibrations to preserve the closed structure of the total category
imposes a technical limitation to the applicability of the categorical
$\top\top$-lifting.  Indeed, outside the categorical semantics of type
theories, it is common to work with the categories that are not
closed. In the study of coalgebras, {\em predicate /
  relational liftings} of functors and monads are fundamental
structures to formulate modal operators and (bi)simulation relations,
and the underlying categories of them are not necessarily closed. For
instance, the category $\Meas$ of measurable spaces, which is not
cartesian closed, is used to host labelled Markov processes \cite{DBLP:journals/tcs/BreugelMOW05}. The
categorical $\top\top$-lifting does not work in such situations.

To overcome this technical limitation, in this paper we introduce an
alternative lifting method called {\em codensity lifting}. The idea is
to replace the continuation monad $(-\Arrow S)\Arrow S$ with the {\em
  codensity monad} $\Ran SS$ given by a right Kan extension.  We then
ask fibrations to preserve the right Kan extension, which is often
fulfilled when $\EE$ has and $p$ preserves limits. We demonstrate that the
codensity lifting is applicable to lift  monads on the
base categories of the following fibrations:
\begin{displaymath}
  \xymatrix@=1.2pc{
    \Pre \ar[d] & \Top \ar[d] & \ERel(\Meas) \ar[r] \ar[d] \pbmark & \BRel(\Meas) \ar[d] \ar[r] \pbmark & \Pred \ar[d] & U^*\Met \ar[r] \ar[d] \pbmark & \Met \ar[d]
    \\
    \Set & \Set & \Meas \ar[r]_-{\Delta} & \Meas^2 \ar[r]_-{(\times)\circ U^2} & \Set & \Meas \ar[r]_-U & \Set
  }
\end{displaymath}
The description of these fibrations are in order:
\begin{itemize}
\item Functors $\Pre\arrow\Set$ and $\Top\arrow\Set$ are the  forgetful functors
  from the category of preorders and that of topological spaces to $\Set$.
\item Functor $\BRel(\Meas)\arrow\Meas$ is the fibration for binary
  relations between (the carrier sets of) two measurable spaces. Functor
  $\ERel(\Meas)\arrow\Meas$ is the subfibration of
  $\BRel(\Meas)\arrow\Meas$ obtained by restricting objects in
  $\BRel(\Meas)$ to the binary relations over the same measurable
  spaces.
  
\item Functor $\Met\arrow\Set$ is the  forgetful functor from the category of
  extended pseudometric spaces and non-expansive functions between
  them. We then apply the change-of-base to it to overlay extended
  pseudometrics on measurable spaces.  This yields a fibration
  $U^*\Met\arrow\Set$.
\end{itemize}

By taking the categorical dual of the codensity lifting, we obtain the
method to lift comonads on the base category of a cofibration to its
total category. This method, which we call the {\em density lifting}
of comonads, is newly added to the conference version of this paper
\cite{calco15}.  We illustrate two examples of the density lifting of
$\Set$-comonads along the subobject fibration of $\Set$.

Another issue when we have a lifting $\dot\mT$ of a monad $\mT$ is the
liftability of {\em algebraic operations} for $\mT$ to the lifting
$\dot\mT$.  For instance, let $\dot\mT$ be a lifting of the powerset
monad $\mT_p$ on $\Set$ along the  forgetful functor
$p:\Top\arrow\Set$, which is a fibration. A typical algebraic operation for $\mT_p$ is the
union of $A$-indexed families of sets:
\begin{displaymath}
  \union^A_X:A\pitchfork T_pX\arrow T_pX,\quad
  \union^A_X(f)=\bigcup_{a\in A}f(a);
\end{displaymath}
here $\pitchfork$ denotes the {\em power}.
Then the question is whether we can ``lift'' the function $\union^A_X$ to a
continuous function of type $A\pitchfork\dot T(X,\open X)\arrow\dot
T(X,\open X)$ for every topological space $(X,\open X)$.  We show that
the liftability of algebraic operations to the codensity liftings of monads
has a good characterisation in terms of the parameters supplied to the
codensity liftings.

We are also interested in the categorical property of the {\em
  collection} of liftings of a monad $\mT$ (along a limited class of
fibrations).  We show a characterisation of the class of liftings of
$\mc T$ as a limit of a large diagram of partial orders.

\subsection{Related Work}

This paper is the journal version of \cite{calco15}. We add an
elementary introduction to the $\toptop$-lifting and the codensity lifting
(Section \ref{sec:fromtt}), and the section about the density lifting of comonads
(Section \ref{sec:den}).

In the semantics of programming languages based on typed
$\lambda$-calculi, {\em logical predicates} and {\em logical
  relations} \cite{plotkindef} have been extensively used for
establishing (relational) properties of programs. The categorical
analysis of logical relations emerged around the 90's
\cite{mitsce,mareynolds}, and its fibrational account was given by
Hermida \cite{hermidathesis}. In these works, the essence of logical
predicates and relations is identified as predicate / relational
liftings of the categorical structures corresponding to type
constructors; especially Hermida studied the construction of such
liftings in fibred category theory. Liftings of categorical structures
along a functor are later employed in a categorical treatment of
refinement types \cite{mellieszeilberger}.

One of the earliest work that introduced logical relations
(i.e., relational liftings) for monads is Filinski's PhD thesis
\cite{filinskiphd}. They play a central role in establishing
relationships between two monadic semantics of programming languages
with computational effects
\cite{filinskirepre,wandvaillancourt,filinskicompare,filinskistrovring,katsumatarelating}.
Larrecq, Lasota and Nowak gave a systematic method to lift monads
based on mono-factorisation systems \cite{lrmontyp}. 
Their method is fundamentally different
from the codensity lifting, and the relationship between these lifting
methods is still not clear.

The origin of the codensity lifting goes back to the {\em biorthogonality
technique} developed in proof theory. Girard used this technique in
various contexts, such as 1) the phase space semantics of linear logic
\cite{girardlinear}, 2) the proof of the strong normalisation of
cut-elimination in proof nets \cite{girardlinear} and 3) the
definition of types in the geometry of interaction \cite{goi}. Krivine
also used biorthogonally-closed sets of terms and stacks in his
realisability semantics of classical logic
\cite{krivineclassical}. Pitts introduced a similar technique called
$\top\top$-closure operator, and used $\top\top$-closed relations as a
substitute for admissible relations in the operational semantics of a
functional language. Abadi considered a domain-theoretic analogue of
Pitts' $\top\top$-closure operator, and compared $\top\top$-closed
relations and admissible relations \cite{Abaditt}. Lindley and Stark's
leapfrog method extends Pitts' $\top\top$-closure operator to the
construction of logical predicates for monadic types
\cite{lindley,lindleystark}. The first author gave a categorical
analogue of the leapfrog method \cite{katsumatatt}, which constructs a
lifting of a strong monad along a closed-structure preserving
fibration.

In the coalgebraic study of state transition systems and process
calculi, one way to represent a (bi)simulation relation between
coalgebras is to give a relational coalgebra with respect to a
relational lifting of the coalgebra functor
\cite{hermidajacobs,hesselinkthijs,hughesjacobs,similarq}.  When the
coalgebra functor preserves weak pullbacks, {\em Barr extension}
\cite{barrext} is often employed to derive a relational lifting of the
coalgebra functor.  In a recent work
\cite{DBLP:conf/cmcs/SprungerKDH18}, Sprunger et al. introduced the
codensity lifting of {\em $\Set$-endofunctors} along a partial order
fibration over $\Set$ with fibred small meets. The basic lifting
strategy is the same as this paper; we derive a lifting of an
endofunctor by pulling back a codensity monad along a canonical
morphism. One notable difference is that the parameter to lift an
endofunctor $F$ refers to the category of $F$-algebras. This is
because we need a substitute for the Kleisli category used in the
codensity lifting of monads.

In Section \ref{sec:example}, we illustrate a couple of definitions of
(bi)simulation relations that can be naturally expressed by the
codensity lifting of monads: one is the definition of simulation
relation on a single labelled Markov process
\cite{DBLP:journals/tcs/BreugelMOW05}, and the other is the definition
of bisimulation relation between two labelled Markov processes
\cite{BBLM2014}.

\subsection{Preliminaries}\label{sec:preliminaries}

We use white bold letters $\BB,\CC,\EE,\cdots$ to range over locally small
categories.  We sometimes identify an object in a category $\CC$ and a functor of
type $1\arrow \CC$.

We do a lot of 2-categorical
calculations in $\CAT$.  To reduce the notational burden, we omit writing the
composition operator $\circ$ between functors, or a functor and a
natural transformation. For instance, for functors $G,F,P,Q$ and a natural
transformation $\alpha:P\arrow Q$, by $G\alpha F$ we mean the natural
transformation $G(\alpha_{FI}):G\circ P\circ F(I)\arrow G\circ Q\circ
F(I)$. We use $\bul$ and $*$ for the vertical and horizontal
compositions of natural transformations, respectively.

Let $A$ be a set and $X$ be an object of a category $\CC$. A {\em power of $X$ by $A$} is a
pair of an object $A\pitchfork X$ and an $A$-indexed family of projection morphisms
$\{\pi_a:A\pitchfork X\arrow X\}_{a\in A}$. They satisfy the following universal
property: for any $A$-indexed family of morphisms $\{f_a:B\arrow X\}_{a\in A}$,
there exists a unique morphism $m:B\arrow A\pitchfork X$ such that $\pi_a\circ m=f_a$
holds for all $a\in A$. Here are some examples of powers:
\begin{enumerate}
\item When $\CC=\Set$, the function space $A\Arrow X$ and the
  evaluation function $\pi_a(f)=f(a)$ give a power of $X$ by $A$.
\item When $\CC$ has small products, the product of $A$-fold copies of
  $X$ and the associated projections give a power of $X$ by $A$.
\item When $\CC$ has powers by $A\in\Set$, any functor
  category $[\DD,\CC]$ also has powers by $A$, which can be given
  pointwisely: $(A\pitchfork F)X=A\pitchfork (FX)$.
\end{enumerate}

A {\em right Kan extension} of $F:\AA\arrow\CC$ along $G:\AA\arrow\DD$
is a pair of a functor $\Ran GF:\DD\arrow\CC$ and a natural
transformation $c:(\Ran GF)G\arrow F$ making the mapping
$\ul{(-)}:[\DD,\CC](H,\Ran GF)\arrow [\AA,\CC](HG,F)$ defined by
\begin{displaymath}
  \ul{(-)}(\alpha)=c\bul(\alpha G)
\end{displaymath}
bijective and natural in $H\in[\DD,\CC]$.
A functor $p:\CC\arrow\CC'$
{\em preserves} a right Kan extension $(\Ran GF,c)$ if $(p(\Ran
GF),pc)$ is a right Kan extension of $pF$ along $G$. Thus for any
right Kan extension $(\Ran G{(pF)},c')$ of $pF$ along $G$, we have $p(\Ran
GF)\simeq\Ran G{(pF)}$ by the universal property.

Let $\mT$ be a monad on a category $\CC$. Its components are denoted
by $(T,\eta,\mu)$. The monad induces the {\em Kleisli lifting}
$(-)^\#:\CC(I,TJ)\arrow\CC(TI,TJ)$ defined by $ f^\#=\mu_J\circ Tf $.
We write $J:\CC\arrow\CC_\mT$ and $K:\CC_{\mc T}\arrow\CC$ for the
left and right adjoint of the Kleisli resolution of $\mT$,
respectively. We also write $\epsilon:JK\arrow\Id_{\CC}$ for the
counit of this adjunction. When $\mT$ is decorated with an extra
symbol, like $\dot\mT$, the same decoration is applied to the
components of $\mT$ and the notation for the Kleisli adjunction, like
$\dot\eta,\dot J,\dot\epsilon$, etc.

For the definition of fibrations and related concepts,
see \cite{jacobscltt}.



\section{From $\toptop$-Lifting to Codensity Lifting}
\label{sec:fromtt}

Before introducing the codensity lifting, we first briefly review its
precursor, the {{\em semantic $\toptop$-lifting}} \cite{katsumatatt}. It is a semantic
analogue of Lindley and Stark's {{\em leapfrog method}} \cite{lindleystark,lindley}, and constructs a
{{\em logical predicate}} for a monad $\mT= (T, \eta, \mu)$ on
$\Set$. Below, by a {{\em predicate}} we mean a pair $(X, I)$ of sets such that
$X \subseteq I$.

The semantic $\toptop$-lifting takes a {\em lifting parameter}, which is a pair
of a set $R$ and a predicate $(S, T R)$. Fix such a parameter. The
semantic $\toptop$-lifting is defined as a mapping of a predicate
$(X, I)$ to the predicate $(\TT T X, T I)$, where $\TT T X$ is
constructed in two steps:
\begin{eqnarray}
  T^{\top} X & = & \{ k \in I \Rightarrow R~|~\fa{x \in X}k (x) \in S \} \nonumber \\
  \TT T X & = & \{ c \in T I~|~\fa{k \in T^{\top} X}k^{\#} (c)
  \in S \} . \label{eq:ttlhs}
\end{eqnarray}

Regarding monads as models of computational effects
{\cite{moggicomputational}}, the above two steps may be intuitively
understood as follows. We think of the parameter $R$ as the type of
return values of {\em continuations} (which corresponds to {{\em stack
    frames}} in operational semantics
{\cite{parpolyopeq,lindleystark}}), and the parameter $S$ as a
specification of ``good computations'' over $R$.  Now let $(X, I)$ be
a predicate, which we regard as a set $I$ of values with a
specification $X$ of ``good values''. Then $T^{\top} X$ collects all
the continuations that send good values to good computations, and
$\TT T X$ collects all the computations over $I$ that yield good
computations when passed to continuations in $T^{\top}
X$. Overall, we regard the semantic $\toptop$-lifting as a process to
collect a set $\TT T X$ of good computations from a given set $X$ of
good values. The semantic $\toptop$-lifting is suitable for the
construction of logical predicates for monadic types \cite[Theorem
3.8]{katsumatatt}.

The semantic $\toptop$-lifting can further be formulated in fibred
category theory. To illustrate this, let us introduce the category
$\Pred$, where an object is a predicate and a morphisms from $(X, I)$
to $(Y, J)$ is a function $f : I \rightarrow J$ that maps elements in
$X$ to those in $Y$. The evident forgetful functor
$p : \Pred \rightarrow \Set$ is a partial order fibration: the inverse image of a
predicate $(X, I)$ along a function $f : J \rightarrow I$ is the
predicate $(\{ j~|~f (j) \in X \}, J)$, which we denote by
$f^{\ast} (X, I)$; see \cite[Chapter 0]{jacobscltt} for more detail. Moreover, the
following facts are known: 1) The category $\Pred$ is cartesian
closed. The following gives exponentials in $\Pred$:
\[ (X, I) \dot{\Rightarrow} (Y, J) = (\{ f~|~ \forall x \in X.f (x)
  \in Y \}, I \Rightarrow J) , \] and they are strictly preserved by
$p$ \cite{hermidathesis}. 2) Monads on $\Set$ are always equipped with
the {{\em bind morphism}}
$\sigma : T X \rightarrow (X \Rightarrow T R) \Rightarrow T R$ given
by $\sigma (c) = \lambda k.k^{\#} (c)$, which is derivable from the {\em
  canonical strength} of monads on $\Set$ \cite{moggicomputational}.
Then both $T^\top X$ and $\TT TX$
can be computed using these categorical facts:
\begin{eqnarray}
  (T^{\top} X, I \Rightarrow T R)& = & (X, I) \dArrow (S, T R) \nonumber \\
  (\TT T X, T R) & = & \sigma^{\ast} (((X, I) \dArrow (S, T R)) \dArrow (S, T R)). \label{eq:ttchar}
\end{eqnarray}
We can also deduce from this characterisation that $\TT T$ extends to
a {\em lifting} (Definition \ref{def:lift}) of the monad $\mT$ along
the fibration $p:\Pred\arrow\Set$. A more conceptual reading of
\eqref{eq:ttchar} is that the semantic $\toptop$-lifting is the
inverse image of the continuation monad on $\Pred$ along $\sigma$, as
depicted in \eqref{eq:ttlift}. The right hand side of
\eqref{eq:ttchar} can be computed in a more general situation where
$p$ is a symmetric monoidal closed fibration of type $\EE\arrow\BB$
\footnote{ That is, $\EE, \BB$ are symmetric monoidal closed, and $p$
  strictly preserves the symmetric monoidal closed structure.  } and
$\mT$ is a strong monad on $\BB$. This is the {\em categorical
  $\toptop$-lifting} in \cite[Section 4]{katsumatatt}.

In this paper, we pursue the lifting method based on an alternative
characterisation of the semantic $\toptop$-lifting. We observe that: 1)
the set $T^{\top} X$ is identical to the homset
$\Pred ((X, I), (S, T R))$, and 2) the universal quantification
$\forall k \in T^{\top} X$ in the definition of $\TT T X$ can be
extracted as the intersection of predicates, which corresponds to the
{{\em fibred meet} of the partial order fibration $p : \Pred \arrow
  \Set$. From these observations, we can characterise $\TT T X$ as the fibred meet
  of inverse images:
\begin{equation}
  \label{eq:predcoden}
  (\TT T X,TR) = \bigwedge_{k \in \Pred ((X, I), (S, T R))} ((p k)^{\#})^{*} (S,TR);
\end{equation}
here $\bigwedge$ stands for the fibred meet.

The above presentation of $\TT T X$ in the language of fibred category
theory leads us to {\em adopt} \eqref{eq:predcoden} as the generalised
definition of $\TT T X$ for any partial order fibration
$p : \EE \arrow \BB$ with fibred
meets 
and monad $\mT$ on $\BB$.  Since $p$ need not to be closed-structure
preserving, the generalised definition makes sense in a wide range of
such fibrations, including forgetful functors from the category of
preorders, topological spaces and metric spaces to $\Set$. However,
unlike the categorical $\toptop$-lifting, the meaning of the
generalised definition is not very clear, because it is a direct
encoding of the right hand side of \eqref{eq:ttlhs} in fibred category
theory. In the next section, we introduce the {\em codensity lifting}
for general fibrations, which has a conceptually clear definition
using the {\em codensity monad}. Then in Proposition \ref{pp:single}
we show that the codensity lifting reduces to the right hand side of
\eqref{eq:predcoden} when the lifting parameter is single, and the
fibration has sufficient limits.

\section{Codensity Lifting of Monads}
\label{sec:codense}

Fix a fibration $p:\EE\arrow\BB$ and a monad $\mT$ on
$\BB$. We formally introduce the main subject of this study, {\em liftings} of $\mT$.
\begin{definition}\label{def:lift}
  A {\em lifting} of $\mT$ (along $p$) is a monad $\dot\mT$
  on $\EE$ such that $p\dot T=Tp,p\dot\eta=\eta p$ and $p\dot\mu=\mu p$.
\end{definition}
We do not require {\em fibredness} on $\dot\mT$. The
codensity lifting is a method to construct a lifting of $\mT$ from
the following data called {\em lifting parameter}.
\begin{definition}\label{def:lp}
  A {\em lifting parameter} (for $\mc T$) is a span
  $\liftparam{\EE}{\BB}{\mc T}{\AA}{R}{S}$
  of functors such that $KR=pS$.
  We say that it is {\em single} if $\AA=1$.
\end{definition}
Any single lifting parameter can be written as $(JR,S)$ for some
$R\in\BB$. We therefore call a pair $(R,S)$ of $R\in\BB$ and $S\in\EE_{TR}$ a single
lifting parameter too. This is the same data used in the single-result
categorical $\toptop$-lifting in \cite{katsumatatt}.

Fix a lifting parameter $\liftparam{\EE}{\BB}{\mc T}{\AA}{R}{S}$.  In
this section we introduce the {\em codensity lifting} under the
assumption that the fibration $p:\EE\arrow\BB$ and the functor $S$ of
the lifting parameter satisfy the following {\em codensity condition}.
\begin{definition}\label{def:cc}
  We say that a fibration $p:\EE\arrow\BB$ and
  a functor $S:\AA\arrow\EE$  satisfy the {\em codensity condition} if
  \begin{enumerate}
  \item a right Kan extension of $S$ along $S$ exists, and
  \item $p:\EE\arrow\BB$ preserves this right Kan extension.
  \end{enumerate}
\end{definition}
We give some sufficient conditions for $(p,S)$ to satisfy the
codensity condition.
\begin{proposition}\label{pp:coden}
  Let $p:\EE\arrow\BB$ be a fibration and $S:\AA\arrow\EE$ be a functor.
  The following are sufficient conditions for $(p,S)$ to satisfy
  the codensity condition:
  \begin{enumerate}
  \item $\EE$ has, and $p$ preserves powers, and $\AA=1$.
  \item $\EE$ has, and $p$ preserves small products, and $\AA$ is small discrete.
  \item $\EE$ has, and $p$ preserves small limits, and $\AA$ is small.
  \item $S$ is a right adjoint.
  \end{enumerate}
\end{proposition}
\begin{proof}
  (1-3) are immediate. (4) Let $P$ be a left adjoint of $S$. Then
  the assignment $F\mapsto FP$ extends to a right Kan extension of $F$
  along $S$. This Kan extension is {\em absolute} \cite[Proposition
  X.7.3]{cwm2}.
\end{proof}
\newcommand{\counit}[1]{c}
\newcommand{\munit}[1]{u}
\newcommand{\multi}[1]{m}

Assume that $(p,S)$ satisfies the codensity condition. We take a right
Kan extension $(\Ran SS,\counit S:(\Ran SS)S\arrow S)$, and equip it
with the following monad structure: the unit
$\munit S:\Id\arrow\Ran SS$ and multiplication
$\multi S:(\Ran SS)(\Ran SS)\arrow\Ran SS$ are respectively unique
natural transformations such that $\counit S\bul \munit SS=\id_S$ and
$\counit S\bul \multi SS=\counit S\bul(\Ran SS)\counit S$.  This is
the {\em codensity monad} \cite[Exercise X.7.3]{cwm2}.

Since $p$ preserves the right Kan extension $\Ran SS$,
$(p(\Ran SS),p\counit S)$ is a right Kan extension of $pS$ along
$S$. Thus the mapping
$\ul{(-)}:[\EE,\BB](H,p(\Ran SS))\arrow [\AA,\BB](HS,pS)$ defined by
\begin{displaymath}
  \ul{(-)}=p\counit S\bul -S
\end{displaymath}
is bijective and natural on $H:\EE\arrow\BB$.  We denote its inverse
by $\ol{(-)}$, and call $\ol f$ the {\em mate} of $f$.

The codensity lifting constructs a lifting $\TT\mT=(\TT T,\TT\eta,\TT\mu)$
of $\mT$ along $p$ as follows.

\subsection*{Lifting the Endofunctor $T$}

We first regard $K\epsilon R$ as a natural transformation of type
$TpS\arrow pS$. We then apply the function $\ol{(-)}$ to it:
\begin{displaymath}
  \infer{\ol{K\epsilon R}:Tp\arrow p(\Ran SS)}{
    K\epsilon R:TpS=KJpS=KJKR\arrow KR=pS}
\end{displaymath}
We next take its cartesian lifting with respect to $\Ran SS$.
\begin{displaymath}
  \xymatrix{
    \TT T  \ar@{.>}[rr]^-{\sigma} & & \Ran SS & [\EE,\EE] \ar[d]^-{[\EE,p]} \\
    Tp \ar[rr]_-{\ol{K\epsilon R}} & & p(\Ran SS) & [\EE,\BB]
  }
\end{displaymath}
This is possible because $[\EE,p]:[\EE,\EE]\arrow[\EE,\BB]$ is again a
fibration. We name the cartesian lifting $\sigma$, and its domain
$\TT T$. We have $p\TT T=[\EE,p]\TT T=Tp$.


\subsection*{Lifting the Unit $\eta$}
Consider the following diagram:
\begin{displaymath}
  \xymatrix{
    \Id_\EE \ar@{.>}[rd]_-{\TT\eta} \ar@/^1pc/[rrrd]^-{\munit S} \\
    & \TT T \ar[rr]_-\sigma & & \Ran SS & [\EE,\EE] \ar[dd]^-{[\EE,p]} \\
    p \ar[rd]_-{\eta p} \ar@/^1pc/[rrrd]^-{p\munit S}  \\
    & Tp \ar[rr]_-{\ol{K\epsilon R}} & & p(\Ran SS) & [\EE,\BB]
  }
\end{displaymath}
The triangle in the base category commutes by:
\begin{displaymath}
  \ol{K\epsilon R}\bul\eta p=
  \ol{K\epsilon R\bul\eta pS}=
  \ol{K\epsilon R\bul\eta KR}=
  \ol{\id_{KR}}=
  \ol{\id_{pS}}=
  p\munit S.
\end{displaymath}
Therefore from the universal property of $\sigma$, we obtain the
unique natural transformation $\TT\eta$ above $\eta p$ making the
triangle in the total category commute.

\subsection*{Lifting the Multiplication $\mu$}

Consider the following diagram.
\begin{displaymath}
  \xymatrix{
    \TT T\TT T \ar@{.>}[rd]_-{\TT\mu} \ar[r]^-{\TT T\sigma} & \TT T\Ran SS \ar[rr]^-{\sigma\Ran SS} & & (\Ran SS)\Ran SS \ar@/^1pc/[rd]^-{\multi S} \\
    & \TT T \ar[rrr]^-\sigma & & & \Ran SS & [\EE,\EE] \ar[dd]^-{[\EE,p]} \\
    TTp \ar[r]^-{T\ol{K\epsilon R}} \ar[rd]_-{\mu p} & Tp(\Ran SS) \ar[rr]^-{\ol{K\epsilon R}\Ran SS} & & p(\Ran SS)\Ran SS \ar@/^1pc/[rd]^-{p\multi S} \\
    & Tp \ar[rrr]_-{\ol{K\epsilon R}} & & & p(\Ran SS) & [\EE,\BB]
  }
\end{displaymath}
The pentagon in the base category commutes by:
\begin{eqnarray*}
  \ul{p \multi S \bul \overline{K\epsilon R} \Ran SS \bul T \ol{K\epsilon R}}
  & = & p \counit S \bul p ( \Ran SS ) \counit S \bul \ol{K\epsilon R} ( \Ran SS ) S \bul T \ol{K\epsilon R} S\\
  (\text{interchange law})& = & p \counit S \bul \ol{K\epsilon R} S \bul T p \counit S \bul T \ol{K\epsilon R} S = K\epsilon R \bul KJ K\epsilon R \\
  & = & K\epsilon R\bul \mu KR =  K\epsilon R\bul \mu pS  =  \ul{\ol{K\epsilon R}\bul \mu p}.
\end{eqnarray*}
Therefore from the universal property of $\sigma$,
we obtain the unique morphism $\TT\mu$ above $\mu p$ making
the pentagon in the total category commute.

\begin{theorem}\label{th:firstlift}
  Let $p:\EE\arrow\BB$ be a fibration, $\mT$ be a monad on $\BB$,
  $\liftparam \EE\BB\mT\AA RS$ be a lifting parameter for $\mT$, and
  assume that $(p,S)$ satisfies the codensity condition. The tuple
  $\TT{\mT}=(\TT T,\TT\eta,\TT\mu)$ constructed as above is a lifting
  of $\mT$ along $p$.
\end{theorem}
\begin{proof}
  From the universal property of the cartesian morphism $\sigma$, it
  suffices to show the following three equalities:
  \begin{displaymath}
    \sigma\bul\TT\mu\bul\TT T\TT\eta=\sigma,\quad
    \sigma\bul\TT\mu\bul\TT\eta\TT T=\sigma,\quad
    \sigma\bul\TT\mu\bul\TT\mu\TT T=\sigma\bul\TT\mu\bul\TT T\TT\mu.
  \end{displaymath}
  They are easily shown from the definition of $\TT\eta$ and
  $\TT\mu$. For instance,
  \[
  \begin{array}[b]{rcl}
    \sigma\bul\TT\mu\bul\TT T\TT\eta & = &
    \multi S\bul\sigma(\Ran SS)\bul\TT T\sigma\bul\TT T\TT\eta=
    \multi S\bul\sigma(\Ran SS)\bul\TT T\munit S \\
    \text{(interchange law)} & = & \multi S\bul(\Ran SS)\munit S\bul\sigma=\sigma.
  \end{array}
  \tag*{\qedhere}
  \]
\end{proof}
\begin{corollary}\label{co:monadmor}
  The cartesian morphism $\sigma:\TT T\arrow\Ran SS$ is a monad morphism.
\end{corollary}

Any lifting of $\mT$ along $p$ can be obtained by the codensity lifting, although
the choice of the lifting parameter is rather canonical.
\begin{theorem}\label{th:comp}
  Let $p:\EE\arrow\BB$ be a fibration, $\mT$ be a monad on $\BB$ and
  ${\dot\mT}$ be a lifting of $\mT$.
  Then there exists a lifting parameter $R,S$ such that $(p,S)$
  satisfies the codensity condition and $\dot\mT\simeq \TT \mT$.
\end{theorem}
\begin{proof}
  We write $p_k:\EE_{{\dot\mT}}\arrow\BB_\mT$ for the canonical
  functor extending $p:\EE\arrow\BB$ to Kleisli categories. Then the
  span $\liftparam{\EE}{\BB}{\mT}{\EE_{\dot\mT}}{p_k}{\dot K}$ is a
  lifting parameter that satisfies the codensity condition by
  Proposition \ref{pp:coden}.
  Since $\dot T=\dot K\dot J$, $(\dot T,\dot K\dot \epsilon)$ is a right Kan extension of $\dot K$ along $\dot K$,
  and this is preserved by $p$.
  Moreover, the morphism
  $\ol{K\epsilon p_k}:Tp\arrow p\dot T=Tp$ becomes the identity
  morphism. Hence ${\dot\mT}$ is isomorphic to $\TT\mT$.
\end{proof}

The codensity lifting is given with respect to the Kleisli resolution
$J\dashv K:\CC_\mT\arrow\CC$ of $\mT$. In fact, we can replace it with
the Eilenberg-Moore resolution $J'\dashv K':\CC^\mT\arrow\CC$ of
$\mT$, because the initiality of the Kleisli resolution of $\mT$ is
irrelevant in the codensity lifting. This replacement affects the
argument in this section as follows:
\begin{itemize}
\item A lifting parameter becomes a pair
  $\liftparamem{\EE}{\BB}{\mc T}{\AA}{R}{S}$ of functors $R,S$ such
  that $pS=K'R$. The codensity condition remains the same.
\item When lifting the components of $\mT$,
  we replace $J$ with $J'$, $K$ with $K'$, and $\epsilon$ with the counit $\epsilon'$
  of $J'\dashv K'$.
\item Theorem \ref{th:firstlift} and Corollary \ref{co:monadmor}
  remains the same.
\item In the proof of Theorem \ref{th:comp}, we use the Eilenberg-Moore
  resolution $\EE^{{\dot\mT}}$ of $\dot\mT$ instead of
  $\EE_{{\dot\mT}}$. The remaining part is the same.
\end{itemize}

At this moment we do not know which resolution is better for the
codensity lifting. Theorem \ref{th:comp} shows that for deriving any
lifting of $\mT$, it is enough to use $\CC_\mT$-valued lifting
parameters; enlarging $\CC_\mT$ to $\CC^\mT$ does not increase the
expressiveness of the codensity lifting. 
In this paper we use the Kleisli
resolution in the codensity lifting.


\section{Examples of Codensity Liftings with Single Lifting Parameters}
\label{sec:example}

We illustrate codensity liftings of monads with single lifting
parameters where 1) the fibration $p:\EE\arrow\BB$ has fibred small
products and 2) $\BB$ has small products. In this situation, $\EE$
also has small products that are preserved by $p$ \cite[Exercise
9.2.4]{jacobscltt}, and any single
lifting parameter satisfies the codensity condition.  We below give a
formula to compute the codensity lifting of a monad with a
single lifting parameters.
\begin{proposition}\label{pp:single}
  Let $p:\EE\arrow\BB$ a fibration such that $p$ has fibred small
  products and $\BB$ has small products, and let $\mc T$ be a monad on
  $\BB$.  Then the functor part of the codensity lifting $\TT\mT$ of
  $\mT$ with a single lifting parameter $R\in\BB$ and $S\in\EE_{TR}$
  satisfies
  \begin{equation}
    \label{eq:general}
    \TT TX\simeq\bigwedge_{f\in\EE(X,S)}((pf)^\#)^{-1}(S),
  \end{equation}
  where $\bigwedge$ stands for the fibred product in $\EE_{T(pX)}$.
\end{proposition}
\begin{proof}
  We supply the lifting parameter $(JR,S)$ to the codensity liftinig (see the convention after
  Definition \ref{def:lp}).
  Let $X\in\EE$. We take the power $(\EE(X,S)\pitchfork pS,\pi)$ in
  $\BB$.  From \cite[Exercise 9.2.4]{jacobscltt}, the object
  \begin{displaymath}
    \bigwedge_{f\in\EE(X,S)}(\pi_f)^{-1}(S)
  \end{displaymath}
  together with an appropriate projection morphism is a
  power of $S$ by $\EE(X,S)$, and $p$ sends it to the power
  $(\EE(X,S)\pitchfork pS,\pi)$. Therefore the pair
  $(\EE(-,S)\pitchfork pS,\pi_{\id_S})$ is a right Kan extension of $pS$ along $S$, and
  the mate function $\ol{(-)}:\EE(FS,pS)\arrow[\EE,\BB](F,\EE(-,S)\pitchfork pS)$ of this
  right Kan extension is given by
  \begin{displaymath}
    (\ol f)_X=\angle{f\circ Fg}_{g\in\EE(X,S)}:FX\arrow \EE(X,S)\pitchfork pS.
  \end{displaymath}
  From this, we have $ (\ol{K\epsilon JR})_X=\angle{K\epsilon_{JR}\circ
    KJpg}_{g\in\EE(X,S)}=\angle{(pg)^\#}_{g\in\EE(X,S)}$. Therefore
  \begin{align*}
    \TT TX & =
                 (\angle{(pg)^\#}_{g\in\EE(X,S)})^{-1}\left(\bigwedge_{f\in\EE(X,S)}(\pi_f)^{-1}(S)\right) \\
           & \simeq
                      \bigwedge_{f\in\EE(X,S)}(\angle{(pg)^\#}_{g\in\EE(X,S)})^{-1}(\pi_f)^{-1}(S) \\
           & \simeq
                      \bigwedge_{f\in\EE(X,S)}(\pi_f\circ\angle{(pg)^\#}_{g\in\EE(X,S)})^{-1}(S) \\
           & =
                 \bigwedge_{f\in\EE(X,S)}((pf)^\#)^{-1}(S).
  \tag*{\qedhere}
  \end{align*}
\end{proof}
In the rest of this section, we instantiate the parameters of
Proposition \ref{pp:single} and identify the right hand side of
\eqref{eq:general}. All the fibrations appearing in this section have
fibred small limits, and are over categories with small limits.

\subsection{Lifting $\Set$-Monads to the Category of Preorders}
\label{sec:expre}

The forgetful functor $p:\Pre\arrow\Set$ from the category
$\Pre$ of preorders and monotone functions is a fibration with fibred
small limits. The inverse image of a preorder $(J,\le_J)$ along a
function $f:I\arrow J$ is the preorder $(I,\le_I)$ given by $i\le_I
i'\iff f(i)\le_J f(i')$. The fibred small limits are given by the
set-theoretic intersections of preorder relations. Although
the category
$\Pre$ is cartesian closed, $p$ does not preserve exponentials. Hence
the categorical $\top\top$-lifting \cite{katsumatatt} is not applicable for
lifting $\Set$-monads along $p$.

We consider the codensity lifting of a monad $\mT$ on $\Set$ along
$p:\Pre\arrow\Set$ with a single lifting parameter: a pair of
$R\in\Set$ and $S=(TR,\le)\in\Pre$.  By instantiating
\eqref{eq:general}, for every preorder $(X,\le_X)\in\Pre$ ($X$ for
short), $\TT TX$ is the preorder $(TX,\TT\le_X)$ where $\TT\le_X$ is given by
\begin{equation}
  \label{eq:ttpre}
  x\mathbin{\TT\le_X} y\iff
  \fa{f\in \Pre(X,S)}(pf)^\#(x)\le (pf)^\#(y).
\end{equation}

We further instantiate this by letting $\mc T$ be the powerset monad $\mc
T_p$, $R=1$ and $\le$ be one of the following partial orders on
$T_p1=\{\emptyset,1\}$:
\begin{enumerate}
\item Case ${\le}=\{\emptyset\le 1\}$. The homset
  $\Pre(X,S)$ is isomorphic to the set $\Up(X)$
  of upward closed subsets of $X$, and \eqref{eq:ttpre} is rewritten to:
  \begin{eqnarray*}
    x\mathbin{\TT\le_X} y & \iff &
                                   (\fa{F\in\Up(X)}x\cap F\neq\emptyset\implies y\cap F\neq\emptyset) \\
                          & \iff & \fa{i\in x}\ex{j\in y}i\le_X j,
  \end{eqnarray*}
  that is, $\TT\le_X$ is the lower preorder.
\item Case ${\le}=\{1\le\emptyset\}$. By the similar argument,
  $\mathbin{\TT\le}$ is the upper
  preorder:
  \begin{eqnarray*}
    x\mathbin{\TT\le_X} y & \iff & \fa{j\in y}\ex{i\in x}i\le_X j.
  \end{eqnarray*}
\end{enumerate}

In order to make $\TT\le$ the {\em convex preorder} on $\mc T_p$:
\begin{displaymath}
  x\mathbin{\TT\le_X} y\iff(\fa{i\in x}\ex{j\in y}i\le_X j)\wedge(\fa{j\in y}\ex{i\in x}i\le_X j),
\end{displaymath}
it suffices to supply the cotupling $\Set_{\mc T_p}\leftarrow 1+1\rightarrow\Pre$
of the above two lifting parameters to the codensity lifting.

\subsection{Lifting $\Set$-Monads to the Category of Topological Spaces}
\label{sec:extop}

The forgetful functor $p:\Top\arrow\Set$ from the category $\Top$ of
topological spaces and continuous functions is a fibration with fibred
small limits. For a topological space $(X,\open X)$ and a function
$f:Y\arrow X$, the inverse image topological space $f^{-1}(X,\open X)$ is
given by $(Y,\{f^{-1}(U)~|~U\in \open X\})$.
We note that each fibre category $\Top_X$ on a set $X$
is the poset of topologies on $X$ ordered by the coarseness,
that is, $(X,\open 1)\le (X, \open 2)$ holds if and only if $\open
2\subseteq\open 1$.

We consider the codensity lifting of a monad $\mT$ on $\Set$ along
$p:\Top\arrow\Set$ with a single lifting parameter: a pair of
$R\in\Set$ and  $S=(TR,\open S)\in\Top$. 
By instantiating \eqref{eq:general},
for every $(X,\open X)\in\Top$ ($X$ for short),
$\TT
TX$ is the topological space $(TX,\TT T\open{X})$ whose
topology $\TT T\open{X}$ is the coarsest one making
every set $((pf)^\#)^{-1}(U)$ open,
where $f$ and $U$ range over $\Top(X,S)$ and $\open S$, respectively.

We further instantiate this by letting $\mc T=\mc T_p,R=1$, and
$\open S$ be one of the following topologies on $T_p1$. The resulting
codensity liftings respectively equip $T_pX$ with the same topologies as {\em
  lower} and {\em upper Vietoris topologies}, which appear in the
construction of {\em hyperspace} \cite{nadler}.
\begin{enumerate}
\item Case $\open S=\{\emptyset,\{1\},\{\emptyset,1\}\}$.  The
  topology $\TT T_p\open X$ is the coarsest one making every set
  $\{V\subseteq pX~|~V\cap U\neq\emptyset\}$ open, where $U$ ranges
  over $\open X$. We call this  {\em lower Vietoris
    lifting}.
\item Case $\open S=\{\emptyset,\{\emptyset\},\{\emptyset,1\}\}$.
  The topology $\TT T_p\open X$ is the coarsest one making every set
  $\{V\subseteq pX~|~V\subseteq U\}$ open, where $U$ ranges over
  $\open X$. We call this {\em upper Vietoris lifting}.
\end{enumerate}
We note that the hyperspace of $X\in\Top$ has closed subsets of $X$ as
points, and therefore is {\em not} a lifting of the powerset monad $\mc T_p$.

\subsection{Simulations on LMPs by Codensity Lifting}
\label{sec:exmeas}

We next move on to the category $\Meas$ of measurable spaces and
measurable functions between them. Recall that $\Meas$ has small
limits. We introduce some notations: For $X\in\Meas$, by $\meas X$ we
mean the $\sigma$-algebra of $X$. For a topological space $X\in\Top$,
by $\mc BX\in\Meas$ we mean the Borel space associated to $X$.

Let $p:\Pred\arrow\Set$ be the subobject fibration of $\Set$. We here
explicitly give $\Pred$ as follows: an object of $\Pred$ is a pair of
sets $(X,I)$ such that $X$ is a subset of $I$, and a morphism from
$(X,I)$ to $(Y,J)$ is a function $f:I\arrow J$ such that
$f(X)\subseteq Y$. The first and second component of $X\in\Pred$ is
denoted by $X_0$ and $X_1$, respectively.  The fibration $p$ has
fibred small limits.
We  then consider the following two fibrations $q,r$ obtained by the
change-of-base of the subobject fibration $p:\Pred\arrow\Set$:
\begin{displaymath}
  \xymatrix{
    \ERel(\Meas) \ar[r] \ar[d]_-r \pbmark & \BRel(\Meas) \ar[rr] \ar[d]_-q \pbmark & & \Pred \ar[d]^-p \\
    \Meas \ar[r]_-{\Delta} & \Meas^2 \ar[r]_-{U^2} & \Set^2 \ar[r]_-{\Rm{Prod}} \ar[r] & \Set
  }
\end{displaymath}
Here, $\Delta$ is the diagonal functor and $\Rm{Prod}$ is the binary product
functor. The derived legs $q$ and $r$ are again fibrations with fibred
small limits.\footnote{$\BRel$ and $\ERel$ stand for binary
  relations and endo-relations, respectively.} An explicit description
of $\BRel(\Meas)$ is:
\begin{itemize}
\item An object $\BX$ is a triple,
  whose components are denoted by $\BX_0,\BX_1,\BX_2$,
  such that $\BX_1,\BX_2$ are measurable spaces and
  $\BX_0\subseteq U\BX_1\times U\BX_2$ is a binary
  relation between the carrier sets of $\BX_1$ and $\BX_2$.

\item A morphism $(f_1,f_2):\BX\arrow\BY$ is a pair
  of measurable functions $f_1:\BX_1\arrow\BY_1$ and
  $f_2:\BX_2\arrow\BY_2$ such that $(Uf_1\times Uf_2)(\BX_0)\subseteq
  \BY_0$.
\end{itemize}
An explicit definition of $\ERel(\Meas)$ is:
\begin{itemize}
\item An object $\BX$ is a pair,
  whose components are denoted by $\BX_0,\BX_1$,
  such that $\BX_1$ is a measurable space and
  $\BX_0\subseteq U\BX_1\times U\BX_1$ is a binary
  relation on the carrier set of $\BX_1$.

\item A morphism $f:\BX\arrow\BY$ is a
  measurable function $f:\BX_1\arrow\BY_1$
  such that $Uf_1(\BX_0)\subseteq \BY_0$.
\end{itemize}

Before proceeding further, we introduce some concepts and notations about
binary relations. For a binary relation $R\subseteq X\times Y$ and a subset
$A\subseteq X$, the image of $A$ by $R$ is defined to be the set
$\{y\in Y~|~\ex{x\in A}(x,y)\in R\}$, and is denoted by $R[A]$. For a
subset $V\subseteq I$, by $\chi_V:I\arrow [0,1]$ we mean the indicator
function defined by: $\chi_V(i)=1$ when $i\in V$ and $\chi_V(i)=0$
when $i\not\in I$.  For two binary relations $P\subseteq I\times J$
and $Q\subseteq J\times K$, by $P;Q\subseteq I\times K$ we mean the
relational composition of $P$ followed by $Q$.  We extend this
operation to $\BRel(\Meas)$-objects $X,Y$ such that $X_2=Y_1$ in the
evident way.

\newcommand{\SPMsr}{\Bf{SPMsr}} The target of the codensity lifting in
this section is the {\em sub-Giry monad} \cite{gilly}, which
we recall below.  For a measurable space $X\in\Meas$,
by $\SPMsr(X)$ we mean the set of sub-probability measures on $X$.  We
equip it with the $\sigma$-algebra generated from the sets of the
following form:
\begin{displaymath}
  \{\mu\in\SPMsr(X)~|~\mu(V)\in W\}\quad(V\in \meas X,W\in\open{[0,1]}),
\end{displaymath}
and denote this measurable space by $GX$. The assignment $X\mapsto GX$
can be extended to a monad $\mG$ on $\Meas$, called the {\em sub-Giry
  monad} \cite{gilly}. Notice that $G1$ is the Borel space
$\mc B[0,1]$ associated to the unit interval $[0,1]$ with
the subspace topology induced from the real line.

\subsubsection{Liftings of Sub-Giry Monad to $\BRel(\Meas)$ and $\ERel(\Meas)$}
We first consider the codensity lifting of the product sub-Giry monad
$\mG^2$ along the fibration $q:\BRel(\Meas)\arrow\Meas^2$ with a
single lifting parameter $R=(1,1)$ (the pair of one-point measurable
space) and $S=(S_0,G1,G1)$; here $S_0$ is a binary relation over the
unit interval $[0,1]=U(G1)$.  From \eqref{eq:general}, the codensity
lifting sends an object $\BX\in\BRel(\Meas)$ to the object $\TT
G\BX\in\BRel(\Meas)$, whose relation part is given by
\begin{displaymath}
  (\TT G\BX)_0=
  \left\{
    (v_1,v_2) ~|~
    \fa{(f,g)\in \BRel(\Meas)(X,S)\hspace{-0.25em}}\hspace{-0.2em}\left(\int_{\BX_1}f~dv_1,\int_{\BX_2}g~dv_2\right) \in S_0
  \right\}.
\end{displaymath}
When the binary relation $S_0$ satisfies certain closure properties, we can
simplify the right hand side of the above equality.
\begin{proposition}\label{prop:subGiryLiftings1}
  Suppose that $S=(S_0,G1,G1)\in\BRel(\Meas)$ satisfies the following conditions:
  \begin{enumerate}
  \item For any object $\BX\in \BRel(\Meas)$, morphism $(f,g):X\arrow
    S$, and $0 \leq \beta \leq 1$, the pair of indicator functions
    $(\chi_{f^{-1}([\beta,1])}, \chi_{g^{-1}([\beta,1])})$ is a
    morphism from $X$ to $S$ in $\BRel(\Meas)$.
  \item The binary relation $S_0$ is closed under taking convex hulls.
  \item The binary relation $S_0$ is closed under taking pointwise suprema.
  \end{enumerate}
  Then the codensity lifting of $\mG^2$ along $q:\BRel(\Meas)\arrow\Meas^2$
  with the single lifting parameter $R=(1,1)$ and $S$ satisfies:
  \[
  (\TT G\BX)_0 =
  \left\{(v_1,v_2)~\vline~
    \begin{array}{l@{}}
      \fa{V \in \meas{\BX_1}, W \in \meas{\BX_2}}\\
      (\chi_V,\chi_W) \in \BRel(\Meas)(X,S)\implies (v_1(V), v_2(W)) \in S_0
    \end{array}
  \right\}.
  \]
\end{proposition}
\begin{proof}
  ($\supseteq$) Obvious.
  ($\subseteq$) Let $(v_1,v_2)$ be a pair in the right hand side relation of the
  above equation.
  Take an arbitrary pair $(f,g)\in \BRel(\Meas)(X,S)$.
  From the definition of Lebesgue integral we obtain
  \[
  \int_{\BX_1}f~dv_1 = \sup\left\{\sum_{n=0}^N \beta_n v_1(f^{-1}([{\textstyle\sum_{i=0}^n \beta_i},1]))~\vline~ \sum_{n=0}^N\beta_n = 1, \beta_n > 0 \right\}
  \]
  From the first and second condition
  \[
  \left(\sum_{n=0}^N \beta_n v_1(f^{-1}([{\textstyle\sum_{i=0}^n \beta_i},1])),\sum_{n=0}^N \beta_n v_2(g^{-1}([{\textstyle\sum_{i=0}^n \beta_i},1]))\right) \in S_0
  \]
  holds for \emph{each} $\{\beta_n\}_{n=0}^N$ such that $\sum_{n=0}^N
  \beta_n = 1$ and $\beta_n > 0$.  From the third condition, we
  conclude $\left(\int_{\BX_1}f~dv_1,\int_{\BX_2}g~dv_2\right) \in
  S_0$.
\end{proof}

We next consider the codensity lifting of $\mG$ along the fibration
$r:\ERel(\Meas)\arrow\Meas$ with $R=1$ and $S=(S_0,G1)$; here
$S_0$ is again  a binary relation over $[0,1]$.  By instantiating
\eqref{eq:general}, we obtain
\begin{displaymath}
  (\TT G\BX)_0
  =
  \left\{
    (v_1,v_2)~|~
    \fa{f\in \ERel(\Meas)(X,S)}\!\left(\int_{\BX_1}f~dv_1,\int_{\BX_1}f~dv_2\right) \in S_0
  \right\}.
\end{displaymath}
The following proposition, which is analogous to Proposition
\ref{prop:subGiryLiftings1}, holds for the above codensity lifting.
\begin{proposition}\label{prop:subGiryLiftings2}
  Suppose that $S=(S_0,G1)\in\ERel(\Meas)$ satisfies the following conditions:
  \begin{enumerate}
  \item For any object $\BX\in\ERel(\Meas)$, morphism $f:X\arrow S$
    and $0 \leq \beta \leq 1$, the indicator function
    $\chi_{f^{-1}([\beta,1])}$ is a morphism from $X$ to $S$ in
    $\ERel(\Meas)$.
  \item The binary relation $S_0$ is closed under taking convex hulls
    and pointwise supremums.
  \end{enumerate}
  Then the codensity lifting $\TT \mG$ of $\mG$ along
  $r:\ERel(\Meas)\arrow\Meas$ with the single lifting parameter $R=1$
  and $S=(S_0,G1)$ satisfies
  \begin{displaymath}
    (\TT G\BX)_0=
    \left\{(v_1,v_2)~|~
      \fa{V \in \meas{\BX_1}} \chi_V \in \BRel(\Meas)(X,S)\implies (v_1(V), v_2(V)) \in S_0
    \right\}.
  \end{displaymath}
\end{proposition}
\subsubsection{Simulations on a Single LMP by Codensity Lifting}

We further instantiate $S_0\subseteq [0,1]^2$ in Proposition
\ref{prop:subGiryLiftings2} with the numerical order $\le$.  The
codensity lifting with this single lifting parameter is simplified as
follows:


\begin{proposition}\label{pp:simuone}
  The codensity lifting $\TT \mG$ of $\mG$ along $r:\ERel(\Meas)\arrow\Meas$
  with the single lifting parameter $R=1$  and $(\le,G1)$
  satisfies:
  \begin{displaymath}
    (v_1,v_2)\in(\TT G\BX)_0\iff
    (\fa{U\in \meas{\BX_1}}\BX_0[U]\subseteq U\implies {v_1}(U)\le{v_2}(U)).
  \end{displaymath}
\end{proposition}
\begin{proof}
  We apply Proposition \ref{prop:subGiryLiftings2}: the binary
  relation $S_0 ={\le}$ is obviously closed under taking convex hulls
  and pointwise supremums, and for any $f \in \ERel(\Meas)(X,S)$ and
  $\beta \in [0,1]$ we have $\chi_{f^{-1}([\beta,1])} \in
  \ERel(\Meas)(X,S)$, because $\beta \leq f(x) \implies \beta \leq
  f(y)$ holds for each $(x,y) \in X_0$.  We then conclude the above
  equivalence because the condition $\BX_0[U]\subseteq U$ is
  equivalent to $\chi_U \in \ERel(\Meas)(X,S)$.
\end{proof}
With the above lifting of sub-Giry monad, we can coalgebraically formulate
{\em simulation relations} on a single {\em labelled Markov process
  (LMP)} proposed in \cite{DBLP:journals/tcs/BreugelMOW05}. Fix a set
$\Act$ of actions. The following is
a coalgebraic definition of LMPs
\cite{DBLP:journals/tcs/BreugelMOW05}:
\begin{definition}
  An LMP is an $(\Act\pitchfork G-)$-coalgebra in $\Meas$.
\end{definition}
We omit the proof of the equivalence between this coalgebraic
definition of LMPs and \cite[Definition
1]{DBLP:journals/tcs/BreugelMOW05}. The concept of simulation relation
on an LMP is proposed in \cite{DBLP:journals/tcs/BreugelMOW05}:
\begin{defiC}[{\cite[Definition 3]{DBLP:journals/tcs/BreugelMOW05}}]
  Let $(X,x)$ be an LMP and $R\subseteq UX\times
  UX$ be a reflexive relation. We say that $R$ is a {\em simulation relation}
  on $(X,x)$ if
  \begin{displaymath}
    \fa{(s_1,s_2)\in R}
    \fa{a\in \Act}
    \fa{U\in\meas X}
    R[U]=U\implies
    \pi_a(x(s_1))(U)\le \pi_a(x(s_2))(U).
  \end{displaymath}
\end{defiC}
In this definition, the formula after ``$\forall{U\in\meas X}$'' is
similar to the right hand side of the equivalence proved in Proposition
\ref{pp:simuone}. Actually, as $R$ is assumed to be reflexive,
$R[U]=U$ is equivalent to $R[U]\subseteq U$. Then the above formula defining
simulation relations can be folded into the existence of a coalgebra in
$\ERel(\Meas)$:
\begin{theorem}\label{th:lmpsim}
  Let $(X,x)$ be an LMP and $R\subseteq UX\times UX$ be a reflexive
  relation. The following are equivalent:
  \begin{enumerate}
  \item $R$ is a simulation relation on $(X,x)$.
  \item $x$ is a morphism of type $(R,X)\arrow \Act\pitchfork \TT G(R,X)$ in $\ERel(\Meas)$,
    where $\TT G$ is the lifting given in Proposition \ref{pp:simuone}.
  \end{enumerate}
\end{theorem}

\subsubsection{Bisimulations between Two LMPs by Codensity Lifting}

We next instantiate $S_0\subseteq [0,1]^2$ in Proposition
\ref{prop:subGiryLiftings1} with the equality relation ${=}$ on $[0,1]$. We
first introduce an auxiliary concept, which appears in
\cite{BBLM2014}.
\begin{definition}
  Let $R\subseteq I\times J$ be a binary relation and $U\subseteq I$ and
  $V\subseteq J$ be subsets. We say that the pair $(V,W)$ is $R$-closed
  if $\fa{(x,y)\in R}(x\in V)\iff (y\in W)$ holds.
\end{definition}
\begin{lemma}
  Let $\BX\in\BRel(\Meas)$ be an object and $V\subseteq U\BX_1$ and
  $W\subseteq U\BX_2$ be arbitrary subsets. Then the following are
  equivalent:
  \begin{enumerate}
  \item $(V,W)$ is $\BX_0$-closed.
  \item $\BX_0\cap(V\times J)=\BX_0\cap(I\times W)$.
  \item $(\chi_V,\chi_W)\in\BRel(\Meas)(\BX,S)$.
  \end{enumerate}
\end{lemma}
\begin{proposition}\label{th:brellift}
  The codensity lifting $\TT\mG$ of $\mG^2$ along
  $q:\BRel(\Meas)\arrow\Meas^2$ with the single lifting parameter
  $R=1$ and $({=},G1,G1)$ satisfies:
  \begin{displaymath}
    (\TT G\BX)_0=
    \left\{({v_1},{v_2})~|~
      \fa{V\in \meas{\BX_1},W\in \meas{\BX_2}}
      (V,W)\colon \BX_0\text{-closed} \implies {v_1}(V) = {v_2}(W)\right\}.
  \end{displaymath}
\end{proposition}
%
%
We point out a relationship between the above codensity lifting and
the concept of bisimulation relation between two LMPs introduced by
Bacci et al \cite{BBLM2014}.
\begin{defiC}[{\cite[Definition 5]{BBLM2014}}]
  Let $(X_1 , x_1 )$ and $(X_2 , x_2 )$ be two LMPs. A binary relation
  $R \subseteq U X_1 \times U X_2$ is a bisimulation relation if
  the following holds:
  \begin{eqnarray*}
    & & \fa{(s_1 , s_2 )\in R}\fa{a\in \Act}\fa{V \in\meas{X_1} , W \in\meas{X_2}} \\
    & & \quad (V, W): R\text{-closed}\implies\pi_a (x_1(s_1 ))(V ) = \pi_a (x_2(s_2 ))(W ).  
  \end{eqnarray*}
\end{defiC}
By folding the above defining formula with the characterisation of $\TT
G$ given in Theorem \ref{th:brellift}, we obtain the following
coalgebraic reformulation of Bacci et al's bisimulation relation:
\begin{theorem}\label{th:bisim}
  Let $(X_i,x_i)$ be LMPs ($i = 1,2$) and  $R\subseteq
  UX_1\times UX_2$ be a binary relation. Then the following are equivalent:
  \begin{enumerate}
  \item $R$ is a bisimulation relation between $(X_1,x_1)$ and $(X_2,x_2)$.
  \item $(x_1,x_2)$ is a morphism of type $(R,X_1,X_2)\arrow {\Act}
    \pitchfork \TT G(R,X_1,X_2)$ in $\BRel(\Meas)$,
    where $\TT G$ is the lifting given in Proposition \ref{th:brellift}.
  \end{enumerate}
\end{theorem}

\subsubsection{Codensity Lifting of $\mG^2$ by the Inequality Relation}

From Proposition \ref{pp:simuone} and Theorem \ref{th:lmpsim}, we
naturally speculate that the codensity lifting of the product sub-Giry
monad $\mathcal{G}^2$ using the inequality relation $\le$ 
yields the lifting that may be used for the
definition of simulation relations between {\em two} LMPs. Below we
try this, and discuss the problem of the composability of simulation
relations.

We instantiate $S_0\subseteq [0,1]^2$ in Proposition \ref{prop:subGiryLiftings1}
with the numerical order $\le$.
\begin{proposition}\label{pp:simtwo}
  The codensity lifting $\TT\mG$ of $\mG^2$ along
  $q\colon\BRel(\Meas)\to\Meas^2$ with the single lifting parameter
  $R=(1,1)$ and $S=(\le,G1,G1)$ satisfies
  \begin{displaymath}
    (\TT G\BX)_0=
    \left\{({v_1},{v_2})~|~
      (\fa{V\in \meas{\BX_1},W\in \meas{\BX_2}}\BX_0[V]\subseteq W\implies {v_1}(V)\le{v_2}(W)
    \right\}
  \end{displaymath}
\end{proposition}
Following Theorem \ref{th:lmpsim} and
Theorem \ref{th:bisim}, we define simulation relations between
two LMPs as coalgebras in $\BRel(\Meas)$.
\begin{definition}
  We define a {\em simulation relation} from an LMP $(X_1,x_1)$ to an
  LMP $(X_2,x_2)$ to be a binary relation $R\subseteq U\BX_1\times
  U\BX_2$ such that $(x_1,x_2)$ is a morphism of type
  $(R,X_1,X_2)\arrow \Act\pitchfork \TT G(R,X_1,X_2)$ in
  $\BRel(\Meas)$, where $\TT G$ is the lifting given in Proposition \ref{pp:simtwo}.
  We moreover say that $R$ {\em preserves measurable sets} if
  for any measurable set $V\in \meas {X_1}$, we have $R[V]\in\meas {X_2}$.
\end{definition}
Unfolding the definition, $R$ is a simulation relation from
$(X_1,x_1)$ to $(X_2,x_2)$ if and only if the following holds:
\begin{eqnarray*}
  & & \fa{(s_1,s_2)\in R}
      \fa{a\in \Act}
      \fa{V\in \meas{\BX_1},W\in \meas{\BX_2}}\\
  & & \quad R[V]\subseteq W\implies
      \pi_a(x_1(s_1))(V)\le \pi_a(x_2(s_2))(W).
\end{eqnarray*}

However, simulation relations defined as above are {\em not}
closed under the relational composition. A counterexample can
be found when $\Act=1$.
\begin{example}\label{ex:failure}
  Let $A$ and $B$ be the discrete and indiscrete
  spaces over a two-point set $2=\{0,1\}$, respectively.  We define three
  probability measures $v_1, v_3\in GA$ and $v_2\in GB$ by:
  \begin{displaymath}
    v_1(\{0\})=v_1(\{1\}) = 1/2,\quad
    v_2(\{0,1\})=1,\quad
    v_3(\{0\})=1/3,~v_3(\{1\})=2/3.\quad
  \end{displaymath}
  We consider three constant functions $k(v_i)$ from $2$ returning $v_i$
  for $i=1,2,3$. They are clearly measurable functions of the following
  type:
  \begin{displaymath}
    k(v_1):A\arrow GA,\quad
    k(v_2):B\arrow GB,\quad
    k(v_3):A\arrow GA,
  \end{displaymath}
  hence they are LMPs (recall $\Act=1$).  We can then easily check that
  $\Rm{Eq}_2\subseteq 2\times 2$ is a simulation relation from $k(v_1)$ to
  $k(v_2)$, and also from $k(v_2)$ to $k(v_3)$.  However, $\Rm{Eq}_2$ is
  not a simulation relation from $k(v_1)$ to $k(v_3)$.
\end{example}
This problem stems from the fact that $\TT\mG$ given in Proposition
\ref{pp:simtwo} does {\em not} satisfy the following property (which
is seen in the definition of 
{\em relators} in
\cite{similarq} and {\em lax extensions} in \cite{Marti2015880}):
\begin{displaymath}
  \fa{X,Y\in\BRel(\Meas)}X_2=Y_1\implies
  (\TT G X;\TT G Y)_0\subseteq
  (\TT G(X;Y))_0,
\end{displaymath}
This is a sufficient condition for the composability of simulation
relations. Example \ref{ex:failure} is actually constructed using a
counterexample to the above property.

A work-around is to require simulation relations to preserve
measurable sets.
\begin{proposition}
  Let $(X_i,x_i)$ be LMPs for $i=1,2,3$, and $R_1\subseteq UX_1\times
  UX_2$ and $R_2\subseteq UX_2\times UX_3$ be measurable-set
  preserving simulation relations from $(X_1,x_1)$ to $(X_2,x_2)$ and
  $(X_2,x_2)$ to $(X_3,x_3)$, respectively. Then $R_1;R_2$ is a
  measurable-set preserving simulation relation from $(X_1,x_1)$ to
  $(X_3,x_3)$.
\end{proposition}
On the other hand, we do not know whether there is a largest
measurable-set preserving simulation relation. We leave this point to
the future work, and move onto other examples of the codensity
lifting.

\subsection{Kantorovich Metric by Codensity Lifting}
\label{sec:kantoro}

An {\em extended pseudometric space} (we omit ``extended'' hereafter)
is a pair $(X,d)$ of a set $X$ and a {\em pseudometric}
$d:X\times X\arrow[0,\infty]$ taking values in the extended
nonnegative real numbers. The pseudometric should satisfy
\begin{displaymath}
  d(x,x)=0,\quad d(x,y)=d(y,x),\quad
  d(x,y)+d(y,z)\ge d(x,z).
\end{displaymath}
For pseudometric spaces $(X,d)$ and $(Y,e)$, a function $f:X\arrow Y$ is
{\em non-expansive} if for any $x,x'\in X$, $d(x,x')\ge e(f(x),f(x'))$ holds.
We define $\Met$ to be the category of
pseudometric spaces and non-expansive functions.
The forgetful functor $p:\Met\arrow\Set$ is a
fibration with fibred small limits. The inverse image of a
pseudometric $(Y,d)$ along a function $f:X\arrow Y$ is given by $f^{-1}(Y,d)=(X,d\circ
(f\times f))$. The fibred small limit of pseudometric spaces
$\{(X,d_i)\}_{i\in I}$ above the same set $X$ is given by the pointwise sup of pseudometrics:
$
\bigwedge_{i\in I}(X,d_i)=(X,\sup_{i\in I}d_i).
$

We first consider the codensity lifting of a monad $\mc T$ on $\Set$ along
$p:\Met\arrow\Set$ with a single lifting parameter: a pair of $R\in\Set$
and $S=(TR,s)\in\Met$. By instantiating \eqref{eq:general}, for every
$(X,d)\in\Met$ ($X$ for short), the pseudometric space
$\TT TX$ is of the form $(TX,\TT Td)$ where the pseudometric $\TT Td$
is given by
\begin{displaymath}
  \TT Td(c,c')=\sup_{f\in\Met(X,S)}s(f^\#(c),f^\#(c')).
\end{displaymath}

We next derive the {\em Kantorovich metric} \cite{kantorovich} on subprobability
measures by the codensity lifting.
We perform the following change-of-base of the fibration
\begin{displaymath}
  \xymatrix{
    U^*(\Met) \ar[d]_-q \ar[rr] \pbmark & & \Met \ar[d]^-p \\
    \Meas \ar[rr]_-U & & \Set
  }
\end{displaymath}
and obtain a new fibration $q$ with fibred small limits.
An object in $U^*(\Met)$ is a pair of a measurable space $(X,\meas X)$
and a pseudometric $d$ on $X$. A morphism from $((X,\meas X),d)$ to
$((Y,\meas Y),e)$ in $U^*(\Met)$  is a measurable function $f:(X,\meas X)\arrow (Y,\meas Y)$ that
is also non-expansive with respect to pseudometrics $d$ and $e$.

We consider the codensity lifting of $\mc G$ along
$q:U^*\Met\arrow\Meas$ with the following single lifting parameter:
the pair of $R=1$ and $S=(G1,s)=(\mc B[0,1],s)$, where $s(x,y)=|x-y|$.
By instantiating \eqref{eq:general}, for every $(X,d)\in\Met$ ($X$ for
short), $\TT GX$ is the pair of the measurable space $GX$ and the
following pseudometric $\TT Gd$ on the set $\SPMsr(X)$ of
subprobability measures on $X$:
\begin{displaymath}
  \TT Gd(v_1,v_2)=\sup_f\left|\int_X fdv_1-\int_X fdv_2\right|;
\end{displaymath}
in the above sup, $f$ ranges over $U^*\Met(X,S)$, the
set of measurable functions of type $X\arrow\mc B[0,1]$ that are also
non-expansive, that is, $\fa{x,y\in UX}d(x,y)\ge |f(x)-f(y)|$. 
This pseudometric $\TT Gd$ between subprobability measures is called
the {\em Kantorovich metric} \cite{kantorovich}.

\section{Density Lifting of Comonads}
\label{sec:den}

The categorical dual of the codensity lifting of monads along
fibrations is the {\em density lifting of comonads} along
cofibrations.

Fix a cofibration $p:\EE\arrow\BB$ and a comonad
$\mD=(D,\epsilon,\delta)$ on $\BB$. We take the co-Kleisli resolution
$(K\dashv J,\eta)$ of the comonad $\mD$.
A {\em lifting parameter} for $\mD$ is a span of functors
$\liftparam{\EE}{\BB}{\mD}{\AA}{R}{S}$ such that $pS=KR$.  We also
assume that $(p,S)$ satisfies the {\em density condition}: $\Lan SS$
exists and $p$ preserves it.

Fix a lifting parameter $\liftparam{\EE}{\BB}{\mD}{\AA}{R}{S}$ and
assume that $(p,S)$ satisfies the density condition. The density
lifting of $\mD$ proceeds as follows. From $pS=KR$, we have the
following natural transformation:
\begin{displaymath}
  K\eta R:pS=KR\arrow KJKR=DKR=DpS.
\end{displaymath}
As $p$ preserves the left Kan extension $\Lan SS$, we obtain a
left Kan extension  $p(\Lan SS)$ of $pS$ along $S$. With this left Kan extension, we
take the mate of the above natural transformation, and obtain
\begin{displaymath}
  \ul{K\eta R}:p(\Lan{S}S)\arrow Dp.
\end{displaymath}
In the cofibration $[\EE,p]:[\EE,\EE]\arrow[\EE,\BB]$, we take the
co-cartesian lifting of this natural transformation with respect to
$\Lan SS$:
\begin{displaymath}
  \xymatrix{
    \Lan SS \ar@{.>}[rr] & & \TT D & [\EE,\EE] \ar[d]^-{[\EE,p]} \\
    p(\Lan SS) \ar[rr] & & Dp & [\EE,\BB]
  }
\end{displaymath}
It yields an endofunctor $\TT D$ above $Dp$, that is, a lifting of the
endofunctor $D$. The counit and comultiplication for $\TT D$ can be
dually constructed as done in Section \ref{sec:codense}.

Let us see the density lifting of a comonad $\mD$ on $\Set$ along the
subobject cofibration $p:\Pred\arrow\Set$ 
with a single lifting parameter $R\in\Set$ and
$S\in\Pred_{DR}$.  It yields a comonad $\TT\mD$ whose functor
part is given by
\begin{equation}
  \label{eq:comontt}
  \TT DX=(\{(pf)^\flat(x)~|~f\in\Pred(S,X),x\in S_0\},DI)\quad
  (X\in\Pred);
\end{equation}
here $(-)^\flat$ is the co-Kleisli lifting of the comonad $\mD$.
Below we instantiate $\mD$ with the {\em product comonad} and
the {\em stream comonad} \cite{uustaluvene}.

\subsection{Density Lifting of Product Comonad}

Fix a set $A$. We consider the product comonad $\mD_A$ on
$\Set$, whose functor part is given by $D_AI=I\times A$.  The
co-Kleisli lifting of this comonad extends a function $f:D_AI\arrow J$
to the function $f^\flat:D_AI\arrow D_AJ$ given by
$f^\flat(i,a)=(f(i,a),a)$.

We instantiate the density lifting \eqref{eq:comontt} with the
product comonad, and obtain
\begin{displaymath}
  \TT D_AX=
  (\{(f(i,a),a)~|~f\in\Pred(S,X),(i,a)\in S_0\},X_1\times A).
\end{displaymath}
This density lifting is actually a product comonad on
$\Pred$.
\begin{theorem}
  Let $R\in\Set$ and $S\in\Pred_{R\times A}$ be a single lifting parameter
  for the product comonad $\mD_A$.  Then the density lifting of
  $\mD_A$ satisfies
  \begin{displaymath}
    \TT D_AX=X\dtimes (S_0[R],A)\quad (X\in\Pred),
  \end{displaymath}
  where $\dtimes$ is the binary product in $\Pred$ given by
  $(X,I)\dtimes(Y,J)= (X\times Y,I\times J)$.
\end{theorem}
\begin{proof}
  ($\subseteq$) easy. ($\supseteq$) Let $(i,a)\in X\times S_0[R]$. There
  exists $r\in R$ such that $(r,a)\in S_0$. We take the constant
  function $k_i:D_AR\arrow I$ returning $i$. This belongs to
  $\Pred(S,X)$ as $i\in X_0$.  Then we obtain
  $(k_i(r,a),a)=(i,a)\in (\TT D_AX)_0$.
\end{proof}

\subsection{Density Lifting of Stream Comonad}

We next consider the stream comonad $\mD$ on $\Set$. Its
functor part sends a set $I$ to the function space $\NN\Arrow I$
from the set $\NN$ of natural numbers.
We regard functions in the space as infinite sequences of elements in
$I$. For an infinite sequence $x\in\NN\Arrow I$ and a natural number
$i\in \NN$, by $x/i$ we mean the infinite sequence
$x_i,x_{i+1},\cdots$, that is, $x/i=\lam jx(i+j)$. The counit and the
comultiplication of the stream comonad are given by
\begin{displaymath}
  \epsilon_I(l)=l(0),\quad \delta_I(l)(m)=l/m.
\end{displaymath}

We instantiate the density lifting \eqref{eq:comontt} with the
stream comonad, and obtain
\begin{displaymath}
  \TT DX=(\{\lam{m}f(s/m)~|~f\in\Pred(S,X),s\in S_0\},\NN\Arrow X_1)\quad
  (X\in\Pred).
\end{displaymath}
We note that $\TT D(\emptyset,I)=(\emptyset,DI)$.
\begin{theorem}
  Let $R\in\Set$ and $S\in\Pred_{\NN\Arrow R}$
  be a single lifting parameter for the stream comonad $\mD$.
  For any $X\in\Pred$, we have the following equivalence:
  \begin{eqnarray*}
    x\in(\TT DX)_0
    & \iff &
    \ex{v\in S_0}
    (\fa{i\in\NN}v/i\in S_0\implies x(i)\in X_0)\wedge \\
    & & \quad\quad\quad\quad(\fa{m,n\in\NN}v/n=v/m\implies x(n)=x(m)).
  \end{eqnarray*}
\end{theorem}
\begin{proof}
It is easy to check that this equivalence holds when $X_0=\emptyset$.
We thus show this equivalence under the assumption that $X_0\neq\emptyset$.

($\Longrightarrow$) 
Take $f\in\Pred(S,X)$ and $s\in S_0$ and assume $x=\lam mf(s/m)$.
We show that $s$ satisfies the two conditions on the right hand side:
\begin{enumerate}
\item Let $i\in\NN$ and assume $s/i\in S_0$. From
  $f\in\Pred(S,X)$, we have $f(s/i)=x(i)\in X_0$.
\item Let $m,n\in\NN$ and assume $s/n=s/m$.
  Then $x(n)=f(s/n)=f(s/m)=x(m)$.
\end{enumerate}

($\Longleftarrow$) Let $x\in DI$ and $v\in S_0$. We consider the following
binary relation $F\subseteq DR\times I$:
\begin{displaymath}
  F=\{(v/i,x(i))~|~i\in\NN\}.
\end{displaymath}
From the condition $\fa{n,m\in\NN}v/n=v/m\implies x(n)=x(m)$, this
binary relation is actually a (graph of a) partial function from $DR$ to
$I$.  Moreover, for any $a\in S_0$, if $F(a)$ is defined, then $F(a)\in
X_0$, because of the condition $\fa{i\in\NN}v/i\in S_0\implies x(i)\in
X_0$.  We now pick an element $y\in X_0$, and extend $F$ to a total
function $F':DR\arrow I$ such that $F'(a)=y$ when $a$ is not in the
domain of $F$. Clearly $F'\in\Pred(S,X)$. Now for any $m\in\NN$,
we have $F'(v/m)=x(m)$. Hence $x\in(\TT DX)_0$.
\end{proof}





\section{Lifting Algebraic Operations to Codensity-Lifted Monads}
\label{sec:}

We introduce the concept of {\em algebraic operation} \cite{DBLP:journals/acs/PlotkinP03}
for general monads, and discuss their liftings to codensity liftings of
monads.  The following definition is a modification  of
\cite[Proposition 2]{DBLP:journals/acs/PlotkinP03} for non-strong
monads, and coincides with the original one when $\CC=\Set$.
\begin{definition}
  Let $\CC$ be a category, $\mT$ be a monad on $\CC$, $A$ be a set and
  assume that $\CC$ has powers by $A$.  An {\em $A$-ary algebraic
    operation} for $\mc T$ is a natural transformation
  $\alpha:A\pitchfork K\arrow K$, where $K$ is the right adjoint of
  the Kleisli resolution of $\mT$.  We write $\Alg(\mT,A)$ for the
  class of $A$-ary algebraic operations for $\mT$.
\end{definition}
\begin{example}\label{ex:union}
  For each set $A$,
  the powerset monad $\mT_p$ has the algebraic operation of {\em
    $A$-ary set-union}
  $
    \union^A_X:A\pitchfork T_pX\arrow T_pX
    $ given by
    $\union^A_X(V)=\bigcup_{a\in A}V_a
  $.
\end{example}

Fix a fibration $p:\EE\arrow\BB$, a monad $\mT$ on $\BB$,
a set $A$ and assume that $\EE$ has and $p$ preserves powers by $A$.
\begin{definition}
  Let $\dot\mT$ be a lifting of $\mT$ along $p$ and
  $\alpha\in\Alg(\mT,A)$ be an $A$-ary algebraic operation for $\mc
  T$.  A {\em lifting} of $\alpha$ to $\dot\mT$ is an algebraic
  operation $\dot\alpha\in\Alg(\dot\mT,A)$ such that
  $p\dot\alpha=\alpha p_k$; here $p_k:\EE_{\dot\mT}\arrow\BB_\mT$ is
  the canonical extension of $p$ to Kleisli categories.  We write
  $\Alg_\alpha(\dot\mT,A)$ for the class $\{\dot\alpha\in
  \Alg(\dot\mT,A)~|~p\dot\alpha=\alpha p_k\}$ of liftings of $\alpha$
  to $\dot\mT$.
\end{definition}
\begin{example}[Continued from Example \ref{ex:union}] \label{ex:unionlift}
   Let $\dot\mT$ be a lifting
  of $\mT_p$ along $p:\Top\arrow\Set$.  Since $p$ is faithful, there
  is at most one lifting of $\union^A$ to $\dot \mT$. It exists if and only if for
  every $(X,\open X)\in\Top$, 
  $\union_X^A$ is a continuous function of
  type $A\pitchfork\dot T(X,\open X)\arrow\dot T(X,\open X)$.
\end{example}

We give a characterisation of the liftings of algebraic
operations to codensity liftings of monads.  Fix a lifting parameter
$\liftparam{\EE}{\BB}{\mT}{\AA}{R}{S}$
and assume that $(p,S)$ satisfies the codensity condition. We
perform the codensity lifting of $\mT$ along $p$ with the lifting
parameter $(R,S)$, and consider the Kleisli resolution of
$\TT\mT$. The functor $p:\EE\arrow\BB$ extends to
$p_k:\EE_{\TT\mT}\arrow\BB_{\mT}$, and it satisfies
\begin{displaymath}
  p_k\TT J=Jp,\quad 
  p\TT K=Kp_k,\quad
  p\TT\eta=\eta p,\quad
  p_k\TT\epsilon=\epsilon p_k.
\end{displaymath}

Starting from a natural transformation $\alpha_0:A\pitchfork S\arrow
S$ such that $p\alpha_0=\alpha R$, we construct a lifting
$\phi(\alpha_0)\in\Alg_\alpha(\TT\mT,A)$ of $\alpha$ as follows.

From $A\pitchfork S=(A\pitchfork\Id_\EE)S$, the natural transformation
$\alpha_0$ induces the mate $\ol{\alpha_0}:A\pitchfork\Id_\EE\arrow
\Ran SS$.
We then obtain the following situation:
\begin{displaymath}
  \xymatrix{
    A\pitchfork\Id_\EE \ar@/^1pc/[rrrd]^-{\overline{\alpha_0}} \ar@{.>}[rd]_-\beta \\
    & \TT T \ar[rr]_-\sigma & & \Ran SS & [\EE,\EE] \ar[dd]^-{[\EE,p]} \\
    A\pitchfork p \ar[rd]_-{\alpha Jp\bul A\pitchfork\eta p} \ar@/^1pc/[rrrd]^-{\overline{\alpha R}} \\
    & Tp \ar[rr]_-{\ol{K\epsilon R}} & & p(\Ran SS) & [\EE,\BB]
  }
\end{displaymath}
The triangle in the base category commutes by:
\begin{eqnarray*}
  \ol{K\epsilon R}\bul\alpha Jp\bul A\pitchfork\eta p
  & = & \ol{K\epsilon R\bul\alpha JpS\bul A\pitchfork\eta pS}
  = \ol{K\epsilon R\bul\alpha JKR\bul A\pitchfork\eta KR} \\
  & = & \ol{(K\epsilon\bul\alpha JK\bul A\pitchfork\eta K)R}
  = \ol{(\alpha\bul A\pitchfork K\epsilon\bul A\pitchfork\eta K)R}
  = 
  p\ol{\alpha_0}.
\end{eqnarray*}
We thus obtain the unique morphism $\beta$ above $\alpha Jp\bul
A\pitchfork\eta p$ making the triangle in the total category
commute. Using this $\beta$, we define
$\phi(\alpha_0):A\pitchfork\TT K\arrow\TT K$ by
\begin{displaymath}
  \phi(\alpha_0)=\TT K\TT\epsilon\bul\beta\TT K:A\pitchfork\TT K\arrow\TT K.
\end{displaymath}
This algebraic operation is 
a lifting of $\alpha$ to $\TT T$:
\begin{displaymath}
  p\phi(\alpha_0)=
  p (\TT K\TT\epsilon\bul\beta\TT K)
  = (K\epsilon \bul \alpha JK \bul A \pitchfork \eta K)p_k
  =  (\alpha \bul A\pitchfork K\epsilon\bul A \pitchfork \eta K)p_k
  = \alpha p_k.
\end{displaymath}
The following theorem shows that $\phi$
characterises the class of liftings of $\alpha$ to
the codensity liftings of monads. It is an analogue of
Theorem 11 in \cite{katsumatarelating}, which is stated for
the categorical $\top\top$-lifting.
\begin{theorem}\label{th:alg}
  Let $p:\EE\arrow\BB$ be a fibration, $\mT$ be a monad on $\BB$, and
  $\liftparam{\EE}{\BB}{\mT}{\AA}{R}{S}$ be a lifting parameter, and
  $A$ be a set.  Suppose that $(p,S)$ satisfies the codensity condition,
  and $\EE$ has, and $p$ preserves
  powers by $A$. Then for any $\alpha\in\Alg(\mT,A)$, the mapping
  $\phi$ constructed as above has the following type and is bijective:
  \begin{displaymath}
    \phi:[\AA,\EE]_{\alpha R}(A\pitchfork S,S)
    \arrow
    \Alg_\alpha(\TT\mT,A).
  \end{displaymath}
\end{theorem}
\begin{proof}
  The candidate $\psi$ of the inverse of $\phi$ is explicitly given as
  follows: it maps $\dot\alpha\in \Alg_\alpha(\TT\mT,A)$ to a morphism
  of type $A\pitchfork S\arrow S$ by the following mate:
  \begin{displaymath}
    \infer{\psi(\dot\alpha)=\ul{\sigma\bul\dot\alpha\TT J\bul A\pitchfork\TT\eta}:A\pitchfork S\arrow S}{
      \sigma\bul\dot\alpha\TT J\bul A\pitchfork\TT\eta:A\pitchfork\Id_\EE\arrow\Ran SS
    }
  \end{displaymath}
  We first show $\psi \circ \phi = \id$:
  \[ \psi ( \phi ( \alpha_0 ) ) = \underline{\sigma \bul \TT{\mu} \bul
    \beta \TT T \bul A \pitchfork \TT{\eta}} = \underline{\sigma \bul
    \TT{\mu} \bul \TT T \TT{\eta} \bul \beta} = \alpha_0 . \] We next
  show $\phi \circ \psi = \id$. By definition, $\phi ( \psi (
  \dot{\alpha} ) ) = \TT K\TT\epsilon \bul \beta \TT K$ where $\beta$
  is the unique morphism above $( \alpha J\bul A \pitchfork \eta ) p$
  such that $\sigma \bul \beta = \overline{\psi ( \dot{\alpha} )} =
  \sigma \bul \dot{\alpha}\TT J \bul A \pitchfork \TT{\eta}$. The
  morphism $\dot{\alpha}\TT J \bul A \pitchfork \TT{\eta}$ is exactly
  such one. Therefore
  \begin{align*}
    \phi ( \psi ( \dot{\alpha} ) )
    & = \TT K\TT\epsilon \bul ( \dot{\alpha}\TT J \bul A \pitchfork \TT{\eta} ) \TT K
    = \TT K\TT\epsilon \bul \dot{\alpha}\TT J\TT K \bul A \pitchfork \TT{\eta} \TT K \\
    & = \dot{\alpha} \bul A \pitchfork \TT K\TT{\epsilon} \bul A \pitchfork \TT{\eta} \TT K
    = \dot{\alpha}.
  \tag*{\qedhere}
  \end{align*}
\end{proof}
\begin{example}[Continued from Example \ref{ex:unionlift}]
  We look at liftings of
  $\union^A\in\Alg(\mT_p,A)$ to the codensity liftings of $\mT_p$ along
  $p:\Top\arrow\Set$ with some single lifting parameters.

  Let $R\in\Set$ and $S=(T_pR,\open S)\in\Top$ be a single lifting
  parameter.  Theorem \ref{th:alg} is instantiated to the following
  statement: a lifting of $\union^A$ to $\TT\mT_p$ exists if and only
  if $\union^A_R:A\pitchfork T_pR\arrow T_pR$ is a continuous function
  of type $A\pitchfork S\arrow S$. Here, $A\pitchfork S$ is the
  product of $A$-fold copies of $S$, and its topology
  $\open {A\pitchfork S}$ is generated from all the sets of the form
  $\pi_a^{-1}(U)$, where $a$ and $U$ range over $A$ and $\open S$,
  respectively.  We further instantiate the single lifting parameter
  $(R,S)$ as follows (see Section \ref{sec:extop}):
  \begin{enumerate}
  \item Case $R=1,\open S=\{\emptyset,\{1\},\{\emptyset,1\}\}$.  For
    any set $A$, $\union^A_1$ is a continuous function of type
    $A\pitchfork S\arrow S$ because
$
      (\union^A_1)^{-1}(\{1\})=
      \bigcup_{a\in A}\pi_a^{-1}(\{1\})\in
      \open {A\pitchfork S}.
$
    From Theorem \ref{th:alg}, for any set $A$, $\union^A$ lifts
    to the lower Vietoris lifting $\TT\mT_p$.

  \item Case $R=1,\open S=\{\emptyset,\{\emptyset\},\{\emptyset,1\}\}$.
    For any finite set $A$, $\union^A_1$ is a
    continuous function of type $A\pitchfork S\arrow
    S$ because
$
      (\union^A_1)^{-1}(\{\emptyset\})=
      \bigcap_{a\in A}\pi_a^{-1}(\{\emptyset\})\stackrel *\in\open {A\pitchfork S}.
$
    On the other hand, the membership $\stackrel *\in$ does not hold
    when $A$ is infinite. 
    From Theorem \ref{th:alg}, for any set $A$,
    $\union^A$ lifts to the upper Vietoris lifting $\TT\mT_p$ if
    and only if $A$ is finite.
  \end{enumerate}
\end{example}



\section{Pointwise Codensity Lifting}
\label{sec:without}

Fix a fibration $p:\EE\arrow\BB$, a monad $\mT$ on $\BB$ and a lifting
parameter $\liftparam{\EE}{\BB}{\mc T}{\AA}{R}{S}$.  When $\AA$ is a
large category, or $\BB,\EE$ are not very complete, the right Kan
extension $\Ran SS$ may not exist, hence the codensity lifting in
Section \ref{sec:codense} is not applicable to lift $\mT$. In this
section we introduce an alternative method, called {\em pointwise
  codensity lifting}, that relies on fibred limits of $p$. The trick
is to swap the order of computation: instead of taking the inverse
image after computing $\Ran SS$, we first take the inverse image of
the components of $\Ran SS$, bringing everything inside a fibre, then
compute the right Kan extension as a fibred limit.

We assume that $\AA$ is small (resp. large) and $p$ has fibred small
(resp. large) limits. The pointwise codensity lifting lifts $\mT$ as
follows. What we actually construct below is a Kleisli triple over $\EE$
which corresponds to a lifting of $\mT$.

\subsection*{Lifting Object Assignment}
We first lift $T$ to an object mapping $\dot T:|\EE|\arrow|\EE|$.
Let $X\in\EE$. Consider the following diagram:
\begin{displaymath}
  \xymatrix{
    X\downarrow S \ar@{}[rd]|-{\Arrow{\gamma_X}} \ar[r]^-{\pi_X} \ar[d]_-{!_{X\downarrow S}} & \AA \ar[r]^-R \ar[d]^-S & \BB_\mT \ar@{=}[r] \ar[d]_-{K} \ar@{}[rd]|(.3){\Arrow\epsilon} & \BB_\mT \ar[d]^-K \\
    1 \ar[r]_-X & \EE \ar[r]_-p & \BB \ar[r]_-T  \ar[ru]_-J & \BB
  }
\end{displaymath}
where $(X\downarrow S,\pi_X,!_{X\downarrow S},{\gamma_X})$ is the {\em
  comma category}. The middle square commutes as $R,S$ is a lifting
parameter. We let $\delta_X=K\epsilon R\pi_X\bul Tp{\gamma_X}$ be the
composite natural transformation, and take the inverse image of $S\pi_X$
along $\delta_X$:
\begin{displaymath}
  \xymatrix{
    \delta_X^{-1}(S\pi_X) \ar@{.>}[rr]^-{\ol{\delta_X}(S\pi_X)} & & S\pi_X & [X\downarrow S,\EE] \ar[d]^-{[X\downarrow S,p]} \\
    TpX!_{X\downarrow S} \ar[rr]_-{\delta_X} & & KR\pi_X & [X\downarrow S,\BB]
  }
\end{displaymath}
We obtain a functor $\delta_X^{-1}(S\pi_X):X\downarrow S\arrow\EE$
such that $p\delta_X^{-1}(S\pi_X)=TpX!_{X\downarrow S}$.  We then define
$\TT TX$ by $ \TT TX=\lim(\delta_X^{-1}(S\pi_X)) $, where right hand
side is the fibred limit.  In the following calculations we will use the
vertical projection and the tupling operation of this fibred limit, denoted
by
\begin{eqnarray*}
  & & P_X:(\TT TX)!_{X\downarrow S}\arrow\delta_X^{-1}(S\pi_X), \\
  & & \angle-:[X\downarrow S,\EE]_{f!_{X\downarrow S}}(Y!_{X\downarrow S},\delta_X^{-1}(S\pi_X))\arrow\EE_f(Y,\TT TX)\quad (f\in\EE(Y,TpX)).
\end{eqnarray*}

\subsection*{Lifting the Unit}
We next lift $\eta$.  Consider the following diagram:
\begin{displaymath}
  \xymatrix{
    X!_{X\downarrow S} \ar@/^1pc/[rrrd]^-{\gamma_X} \ar@{.>}[rd]_-{\eta'_X} \\
    & \delta_X^{-1}(S\pi_X) \ar[rr]_-{\ol{\delta_X}(S\pi_X)} & & S\pi_X & [X\downarrow S,\EE] \ar[dd]^-{[X\downarrow S,p]} \\
    pX!_{X\downarrow S} \ar@/^1pc/[rrrd]^-{p{\gamma_X}} \ar[rd]_-{\eta pX!_{X\downarrow S}} \\
    & TpX!_{X\downarrow S} \ar[rr]_-{\delta_X} & & KR\pi_X & [X\downarrow S,\BB]
  }
\end{displaymath}
where the lower triangle commute by:
\begin{displaymath}
  \delta_X\bul\eta p X!_{X\downarrow S}=
  K\epsilon R\pi_X\bul \eta pS\pi_X\bul p{\gamma_X}=
  K\epsilon R\pi_X\bul \eta KR\pi_X\bul p{\gamma_X}=
  p\gamma_X.
\end{displaymath}
Therefore there exists the unique natural transformation $\eta'_X$ above $\eta
pX!_{X\downarrow S}$ making the upper triangle commute.  We define
$\TT\eta_X=\angle{\eta'_X}$,
which is above $\eta pX$.

\subsection*{Lifting the Kleisli lifting}
We finally lift the Kleisli lifting $(-)^\#$ of $\mT$.
Let $g:X\arrow \TT TY$ be a morphism in $\EE$,
and $f=P_Y\bul g!_{Y\downarrow S}:X!_{Y\downarrow S}\arrow\delta_Y^{-1}(S\pi_Y)$ be a morphism, which is  above $pg!_{Y\downarrow S}$ and satisfies $g=\angle
f$.  We obtain the composite natural transformation
$\ol{\delta_Y}(S\pi_Y)\bul f:X!_{Y\downarrow
  S}\arrow\delta_Y^{-1}(S\pi_Y)\arrow S\pi_Y$.  From the universal
property of the comma category, we obtain the unique functor
$M_f:Y\downarrow S\arrow X\downarrow S$ such that $\pi_XM_f=\pi_Y$ and
$\gamma_XM_f=\ol{\delta_Y}(S\pi_Y)\bul f$. We next consider the
following diagram:
\begin{displaymath}
  \xymatrix{
    \delta_X^{-1}(S\pi_X)M_f \ar@{.>}[rd]_-{f^\flat} \ar@/^1pc/[rrrd]^-{\ol{\delta_X}(S\pi_X)M_f} \\
    & \delta_Y^{-1}(S\pi_Y) \ar[rr]_-{\ol{\delta_Y}(S\pi_Y)} & & S\pi_Y & [Y\downarrow S,\EE] \ar[dd]^-{[Y\downarrow S,p]} \\
    TpX!_{Y\downarrow S} \ar@/^1pc/[rrrd]^-{\delta_XM_f} \ar[rd]_-{\mu pY!_{Y\downarrow S}\bul Tpf} \\
    & TpY!_{Y\downarrow S} \ar[rr]_-{\delta_Y} & & KR\pi_Y & [Y\downarrow S,\BB]
  }
\end{displaymath}
where the lower triangle commutes. 
Therefore there exists the unique natural transformation $f^\flat$
above
$
\mu pY!_{Y\downarrow S}\bul Tpf=
\mu pY!_{Y\downarrow S}\bul Tpg!_{Y\downarrow S}=
(pg)^\#!_{Y\downarrow S}  
$
making the upper triangle
commute.
Then we define $g^{\TT\#}=\angle{f^\flat\bul P_X M_f}$, which
is above $(pg)^\#$.
\begin{theorem}\label{th:gentt}
  Let $p:\EE\arrow\BB$ be a fibration with fibred small (resp. large) limits, $\mT$
  be a monad on $\BB$, $\liftparam\EE\BB\mT\AA RS$ be a lifting
  parameter for $\mT$ and assume that $\AA$ is small (resp. large). The tuple $(\TT
  T,\TT\eta,(-)^{\TT\#})$ constructed as above is a Kleisli triple on $\EE$,
  and the corresponding monad is a lifting of $\mc T$.
\end{theorem}
\begin{proof}
  We first show $(\TT\eta_X)^\#=\angle{\eta'_X}^\#=\id$.  The
  composite natural transformation is
  $\ol{\delta_X}(S\pi_X)\bul\eta'_X=\gamma_X$ by definition.
  Therefore $M_f=\Id_{X\downarrow S}$.  Hence $f^\flat$ is also the
  identity morphism.  Therefore the above composite is also the
  identity morphism.

  We next show $f^\#\circ\TT\eta_X=f$:
  \begin{displaymath}
    \angle{f^\flat\bul P M_f}\circ\TT\eta_X
    =
    \angle{f^\flat\bul P M_f\bul\angle{\eta'_X}!_{Y\downarrow S}}
    =
    \angle{f^\flat\bul P M_f\bul\angle{\eta'_X}!_{X\downarrow S}M_f}
    =
    \angle{f^\flat\bul\eta'_XM_f}
    \stackrel *= 
    \angle f.
  \end{displaymath}
  The last equation $\stackrel *=$ holds because the morphisms on both
  sides are above the same morphism:
  \begin{eqnarray*}
    p(f^\flat\bul\eta'_X M_f)
    & = & pf^\flat\bul p\eta'_XM_f= \mu pY!_{Y\downarrow S}\bul Tpf\bul\eta pX!_{X\downarrow S}M_f\\
    & = & \mu pY!_{Y\downarrow S}\bul Tpf\bul\eta pX!_{Y\downarrow S}=\mu pY!_{Y\downarrow S}\bul \eta p\delta_Y^{-1}(S\pi_Y)\bul pf\\
    & = & \mu pY!_{Y\downarrow S}\bul \eta TpY!_{Y\downarrow S}\bul pf=pf
  \end{eqnarray*}
  and are equalised by the cartesian morphism $\ol{\delta_Y}(S\pi_Y):\delta_Y^{-1}(S\pi_Y)\arrow S\pi_Y$:
  \begin{displaymath}
    \ol{\delta_Y}(S\pi_Y)\bul f^\flat\bul\eta'_X M_f=
    \ol{\delta_X}(S\pi_X)M_f\bul\eta'_X M_f=
    \gamma_X M_f=
    \ol{\delta_Y}(S\pi_Y)\bul f.
  \end{displaymath}

  We finally show $(\angle g^\#\circ \angle f)^\#=\angle
  g^\#\circ\angle f^\#$ for $\angle f:X\arrow\TT TY$ and $\angle
  g:Y\arrow\TT TZ$. Let $h = g^{\flat} \bullet f M_g$.
  \begin{enumerate}
  \item We show $M_f M_g = M_h$. From
    \[ \pi_X M_h = \pi_Z = \pi_Y M_g = \pi_X M_f M_g \]
    and
    \[ \gamma_X M_h = \overline{\delta_Z} ( S \pi_Z ) \bullet g'
    \bullet f M_g = \overline{\delta_Y} ( S \pi_Y ) M_g \bullet f
    M_g = ( \overline{\delta_Y} ( S \pi_Y ) \bullet f )
    M_g = \gamma_X M_f M_g, \]
    the universal property of the comma object makes $M_h = M_f M_g$.
  \item We show $h^{\flat} = g^{\flat} \bullet f^{\flat} M_g$. First
    the following calculation shows $p
    h^{\flat} = p ( g^{\flat} \bullet f^{\flat} M_g )$:
    \begin{eqnarray*}
      p h^{\flat} & = & \mu p Z!_{Z \downarrow S} \bullet T p h\\
      & = & \mu p Z!_{Z \downarrow S} \bullet T p g^{\flat} \bullet T p f M_g\\
      & = & \mu p Z!_{Z \downarrow S} \bullet T ( \mu p Z!_{Z \downarrow S}
      \bullet T p g ) \bullet T p f M_g\\
      & = & \mu p Z!_{Z \downarrow S} \bullet T \mu p Z!_{Z \downarrow S} \bullet
      T^2 p g \bullet T p f M_g\\
      & = & \mu p Z!_{Z \downarrow S} \bullet \mu T p Z!_{Z \downarrow S} \bullet
      T^2 p g \bullet T p f M_g\\
      & = & \mu p Z!_{Z \downarrow S} \bullet T p g \bullet \mu p Y !_{Z
        \downarrow S} \bullet T p f M_g\\
      & = & \mu p Z!_{Z \downarrow S} \bullet T p g \bullet \mu p Y !_{Y
        \downarrow S} M_g \bullet T p f M_g .\\
      p ( g^{\flat} \bullet f^{\flat} M_g ) & = & p g^{\flat} \bullet p
      f^{\flat} M_g\\
      & = & \mu p Z!_{Z \downarrow S} \bullet T p g \bullet ( \mu p Y!_{Y
        \downarrow S} \bullet T p f ) M_g\\
      & = & \mu p Z!_{Z \downarrow S} \bullet T p g \bullet \mu p Y!_{Y
        \downarrow S} M_g \bullet T p f M_g .
    \end{eqnarray*}
    Second, the cartesian morphism $\overline{\delta_Z} ( S \pi_Z )$
    equalise 
    $h^{\flat}$ and $g^{\flat} \bullet f^{\flat} M_g$:
    \begin{eqnarray*}
      \overline{\delta_Z} ( S \pi_Z ) \bullet h^{\flat} & = &
      \overline{\delta_X} ( S \pi_X ) M_h\\
      \overline{\delta_Z} ( S \pi_Z ) \bullet g^{\flat} \bullet
      f^{\flat} M_g & = & \overline{\delta_Y} ( S \pi_Y ) M_g \bullet
      f^{\flat} M_g\\
      & = & ( \overline{\delta_Y} ( S \pi_Y ) \bullet f^{\flat}
      ) M_g\\
      & = & \overline{\delta_X} ( S \pi_X ) M_f M_g\\
      & = & \overline{\delta_X} ( S \pi_X ) M_h .
    \end{eqnarray*}
    Therefore $p h^{\flat} = p ( g^{\flat} \bullet f^{\flat} M_g )$.

  \item Finally,
    we show $\langle g \rangle^{\#} \circ \langle f \rangle^{\#}=
    ( \langle g \rangle^{\#} \circ \langle f \rangle
    )^{\#}$.
    \[
    \begin{array}[b]{rcl}
      \langle g \rangle^{\#} \circ \langle f \rangle^{\#} &
      = & \langle g^{\flat} \bullet P_Y M_g \rangle \circ \langle
      f^{\flat} \bullet P_X M_f \rangle\\
      & = & \langle g^{\flat} \bullet P_Y M_g \bullet \langle f^{\flat}
      \bullet P_X M_f \rangle !_{Z \downarrow S} \rangle\\
      & = & \langle g^{\flat} \bullet P_Y M_g \bullet \langle f^{\flat}
      \bullet P_X M_f \rangle !_{Y \downarrow S} M_g \rangle\\
      & = & \langle g^{\flat} \bullet f^{\flat} M_g \bullet P_X M_f M_g
      \rangle\\
      & = & \langle h^{\flat} \bullet P_X M_h \rangle\\
      ( \langle g \rangle^{\#} \circ \langle f \rangle
      )^{\#} & = & ( \langle g^{\flat} \bullet P_Y M_g
      \rangle \circ \langle f \rangle )^{\#}\\
      & = & ( \langle g^{\flat} \bullet P_Y M_g \bullet \langle f
      \rangle !_{Z \downarrow S} \rangle )^{\#}\\
      & = & ( \langle g^{\flat} \bullet P_Y M_g \bullet \langle f
      \rangle !_{Y \downarrow S} M_g \rangle )^{\#}\\
      & = & ( \langle h \rangle )^{\#}\\
      & = & \langle h^{\flat} \bullet P_X M_h \rangle .
    \end{array}
    \tag*{\qedhere}
    \]
  \end{enumerate}
\end{proof}
The pointwise codensity lifting coincides with the codensity lifting
in Section \ref{sec:codense},
provided that $\Ran SS$ and $p(\Ran SS)$ are both pointwise.
\begin{theorem}\label{th:coin}
  Let $p:\EE\arrow\BB$ be a fibration, $\mT$ be a monad on $\BB$ and
  $\liftparam{\EE}{\BB}{\mT}{\AA}{R}{S}$ be a lifting
  parameter. Assume that $p,S$ satisfies the codensity condition, and
  moreover $\Ran SS$ and $p(\Ran SS)$ are both pointwise.  Then
  $((\ol{K\epsilon R})^{-1}(\Ran SS))X\simeq\lim(\delta_X^{-1}(S\pi_X))$.
\end{theorem}
\begin{proof}
  Let $(c_S,\Ran SS)$ be the pointwise right Kan extension.  Because
  its image by $p$ is also assumed to be a pointwise Kan extension,
  the following diagram is a right Kan extension of $pS\pi_X$ along
  $!_{X\downarrow S}$:
  \begin{displaymath}
    \xymatrix{
      X\downarrow S  \ar@{}[rd]|-{\Arrow{\gamma_X}} \ar[r]^-{\pi_X} \ar[d]_-{!_{X\downarrow S}} & \AA \ar[r]^-S \ar[d]^-S \ar@{}[rd]|-{\Arrow\epsilon} & \EE \ar[r]^-p & \BB \\
      1 \ar[r]_-X & \EE \ar@/_1pc/[rru]_-{p\Ran SS} & &
    }
  \end{displaymath}
  That is, the pair $(L,P)=(p(\Ran SS)X, pc_S\pi_X\bul p\Ran
  SS\gamma_X)$ is a limit of $pS\pi_X$.  Then
  \begin{eqnarray*}
    ((\ol{K\epsilon R})^{-1}(\Ran SS))X
    & =  & (\ol{K\epsilon R}X)^{-1}((\Ran SS)X) \\
    (\text{$\Ran SS$ pointwise}) & \simeq  & (\ol{K\epsilon R}X)^{-1}(\lim S\pi_X) \\
    (\text{limits by fibred limits}) & = & (\ol{K\epsilon R}X)^{-1}(\lim P^{-1}(S\pi_X)) \\
    (\text{preservation of fibred limits}) & \simeq & \lim (\ol{K\epsilon R}X!_{X\downarrow S})^{-1}(P^{-1}(S\pi_X)) \\
    & \simeq & \lim (P\bul\ol{K\epsilon R}X!_{X\downarrow S})^{-1}(S\pi_X)
  \end{eqnarray*}
  By expanding $P$, 
  \[
  \begin{array}[b]{rcl}
    P\bul\ol{K\epsilon R}X!_{X\downarrow S} & = &
    pc_S\pi_X\bul p\Ran SS\gamma_X\bul \ol{K\epsilon R}X!_{X\downarrow S}\\
    & = & 
    pc_S\pi_X\bul \ol{K\epsilon R}S\pi_X\bul Tp\gamma_X \\
    & = & 
    (pc_S\bul \ol{K\epsilon R}S)\pi_X\bul Tp\gamma_X \\
    & = &
    K\epsilon R\pi_X\bul Tp\gamma_X \\
    & = &
    \delta_X.
  \end{array}
  \tag*{\qedhere}
  \]
\end{proof}

\section{Characterising the Collection of Liftings as a Limit}
\label{sec:limit}

We give a characterisation of the class of liftings of a monad on the
base category of a {\em posetal} fibration with fibred small
limits. We show that the class of liftings of $\mc T$ is the vertex of
a certain type of limiting cone.

Fix a posetal fibration $p:\EE\arrow\BB$ with fibred small limits and
a monad $\mT$ on $\BB$.  Notice that each fibre actually admits {\em
  large} limits computed by meets.  Since $p$ is posetal, $p$ is
faithful. Without loss of generality, we regard each homset $\EE(X,Y)$
as a subset of $\BB(pX,pY)$.
\begin{definition}
  We define $\Lift ( \mT )$ to be the class of liftings of $\mT$ along $p$.
  We introduce a partial order $\preceq$ on them by
  \begin{displaymath}
    \dot{T} \preceq \dot{T'} \Longleftrightarrow \fa{X \in \EE}
    \dot{T} X \leqslant \dot{T'} X\quad (\text{in}~\EE_{T(pX)}). 
  \end{displaymath}  
\end{definition}
The partially ordered class $(\Lift(\mT),\preceq)$ admits arbitrary
large meets given by the pointwise meet.

We introduce a convenient notation for the codensity liftings of $\mT$.
By $[S]^R$ we mean the pointwise codensity lifting $\TT\mT$ of $\mT$
with a single lifting parameter $R\in\BB$ and $S\in\EE_{TR}$. By expanding
the definition, we have
\begin{displaymath}
  [S]^{R}X=\bigwedge_{f\in\EE(X,S)}(f^\#)^{-1}(S);
\end{displaymath}
see also \eqref{eq:general}.
\begin{definition}
  Let $X\in\EE$ be an object. An object $S \in \EE_{T(pX)}$ is {\em{closed}} with
  respect to $X$ if
  1) $\eta_{pX}\in\EE(X,S)$ and
  2) for any $f\in\EE(X,S)$, we have $f^\#\in\EE(S,S)$.
\end{definition}
\begin{proposition}\label{pp:closed}
  Let $X\in\EE$ be an object. Then an object $S\in\EE_{T(pX)}$ is closed with respect to $X$
  if and only if $S= [S]^{pX}X$.
\end{proposition}
\begin{proof}
  We first show that $\eta_{pX}\in\EE(X,S)$ if and only if $[S]^{pX}X\le
  S$. (only if) We have
  \begin{displaymath}
    [S]^{pX}X \le ((\eta_{pX})^\#)^{-1}(S) = (\id_{T(pX)})^{-1}(S) = S.
  \end{displaymath}
  (if) As $[S]^{pX}$ is a lifting of $T$, we have
  $\eta_{pX}\in\EE(X,[S]^{pX}X)\subseteq\EE(X,S)$.

  We next show that $\fa{f\in\EE(X,S)}f^\#\in\EE(S,S)$ holds
  if and only if $S\le[S]^{pX}X$.
  \begin{displaymath}
    S\le[S]^{pX}X \iff \fa{f\in\EE(X,S)}S\le(f^\#)^{-1}S
     \iff \fa{f\in\EE(X,S)} f^\#\in\EE(S,S).
  \tag*{\qedhere}
  \end{displaymath}
\end{proof}
\begin{definition}
  Let $X\in\EE$ be an object.
  \begin{enumerate}
  \item We define $\Cls(\mT,X)$ to be the set $\{S\in \EE_{T(pX)}~|~S=
    [S]^{pX}X\}$ of closed objects with respect to $X$.
  \item We regard the codensity lifting $[-]^{pX}$
    as a function of 
    type $\Cls(\mc T,X)\arrow \Lift(\mT)$.
  \item We define the monotone function $q_X
    :(\Lift(\mT),\preceq)\arrow(\Cls(\mc
    T,X),\le)$ to be the evaluation of a given lifting at $X$,
    that is,    $
      q_X(\dot T) = \dot TX.    
      $
    Here, the order $\le$ on $\Cls(\mc T,X)$ is the one inherited from $\EE_{T(pX)}$.
  \item We extend the order $\le$ on $\Cls(\mc T,X)$ to the pointwise
    order between parallel pairs of functions into $\Cls(\mc T,X)$.
  \end{enumerate}
\end{definition}

We note that $[-]^{pX}$ cannot be monotone, because its argument is
used both in a positive and a negative way.  Still, we have the following
adjoint-like relationship:
\begin{theorem}\label{th:adj}
  For any $X\in\EE$, we have
  $
    q_X\circ [-]^{pX}=\id_{\Cls(\mT,X)}$ and $\id_{\Lift(\mT)}\preceq [-]^{pX}\circ q_X.
  $
\end{theorem}
\begin{proof}

  We already have $q_X([S]^{pX})=[S]^{pX}X= S$ from the definition of
  $\Cls(\mT,X)$.

  We show $\dot T\preceq[\dot TX]^{pX}$. We have the following equivalence:
  \begin{eqnarray*}
    \dot T\preceq[\dot TX]^{pX}
    & \iff &
    \fa{Y\in\EE,f\in\EE(Y,\dot TX)}\dot TY\le (f^\#)^{-1}(\dot TX) \\
    & \iff & 
    \fa{Y\in\EE,f\in\EE(Y,\dot TX)}f^\#\in\EE(\dot TY,\dot TX)
  \end{eqnarray*}
  and the last line always holds as $\dot T$ is a lifting of $T$.
\end{proof}
We define a function $\phi_{X,Y}:\Cls(\mT,X)\arrow\Cls(\mT,Y)$ by
\begin{displaymath}
  \phi_{X,Y}(S)=q_{Y}\circ[-]^{pX}(S)=[S]^{pX}Y.
\end{displaymath}
Theorem \ref{th:adj} asserts that
$\phi_{X,X}=\id_{\Cls(\mT,X)}$.
Using the second inequality of Theorem \ref{th:adj},
for any $X,Y\in\EE$, we have
\begin{eqnarray}
  & & q_X \le  q_X\circ [-]^{pY}\circ q_Y=\phi_{Y,X}\circ q_Y \label{eq:lax} \\
  & & [S]^{pX}\preceq[[S]^{pX}Y]^{pY}  =[\phi_{X,Y}(S)]^{pY}.
  \label{eq:prec}
\end{eqnarray}

From Theorem \ref{th:adj}, $\dot T$ is a lower bound of the class
$\{[q_X(\dot T)]^{pX}~|~X\in\EE\}$. In fact, $\dot T$ is the {\em
  greatest} lower bound:
\begin{theorem}\label{th:intersection}
  For any lifting $\dot \mT$ of $\mT$, we have $\dot T= \bigwedge_{X\in\EE}[q_X(\dot T)]^{pX}$. \end{theorem}
\begin{proof}
  It suffices to show $\bigwedge_{X\in\EE}[q_X(\dot T)]^{pX}\preceq\dot T$.
  For any $Y\in\EE$,
  \begin{displaymath}
    \bigwedge_{X\in\EE}[q_X(\dot T)]^{pX}Y
    =
    \bigwedge_{X\in\EE}\phi_{X,Y}(q_{X}(\dot T))
    \le
    \phi_{Y,Y}(q_{Y}(\dot T))
    =
    q_Y(\dot T)
    =
    \dot TY.
    \tag*{\qedhere}
  \end{displaymath}
\end{proof}
This theorem also states that any lifting of a monad $\mT$ is an intersection
of a class of single lifting parameter codensity liftings;
see also Theorem \ref{th:comp}. From this, we obtain the following corollary:
\begin{corollary}
  Let $\mc X\subseteq\Lift(\mT)$ be a class of liftings of $\mT$.  If
  1) for any $X\in\EE$ and $S\in\Cls(\mT,X)$, $[S]^{pX}\in\mc X$, and
  2) $\mc X$ is closed under class-size intersection, then $\mc
  X=\Lift(\mT)$.
\end{corollary}
\begin{definition}
  We say that an object $X\in\EE$ is a {\em split subobject} of an object $Y\in\EE$
  (denoted by $X\lhd Y$) if there is a split monomorphism $m:X\arrow
  Y$.
\end{definition}
One easily sees that the binary relation $\lhd$ on $\Obj\EE$ is reflexive and transitive.
We define $\Split(\EE)$ to be the preordered class
$(\Obj\EE,\lhd)$. 
\begin{lemma}\label{lm:spcone}
  Suppose that $X\lhd Y$ holds for objects $X,Y\in\EE$. The following holds:
  \begin{enumerate}
  \item $\phi_{Y,X}\circ q_{Y}= q_X$.
  \item For any object $Z\in\EE$, $\phi_{Y,X}\circ\phi_{Z,Y}=\phi_{Z,X}$.
  \end{enumerate}
\end{lemma}
\begin{proof}
Assume $X\lhd Y$ in $\EE$.
\begin{enumerate}
\item 
  Let $(\dot T,\dot\eta,\dot\mu)\in\Lift(\mT)$. From
  \eqref{eq:lax} we have $q_X(\dot
  T)\le\phi_{Y,X}(q_{Y}(\dot T))$. To show the opposite inequality,
  take a split monomorphism $m:X\arrow Y$ in $\EE$. It comes with
  $e:Y\arrow X$ such that $e\circ m=\id_X$.  We then consider the
  following diagram:
  \begin{displaymath}
    \xymatrix{
      [\dot TY]^{pY}X \ar[d]_-{\le} \\
      ((\dot\eta_{Y}\circ m)^{\dot\#})^{-1}(\dot TY) \ar@{.>}[r] & \dot TY \ar[rr]^-{\dot Te} & & \dot TX & \EE \ar[d]^-p\\
      T(pX) \ar[r]_-{T(pm)} & T(pY) \ar[rr]_-{T(pe)} & & T(pX) & \BB
    }
  \end{displaymath}
  The composite of the morphisms in the base category is $T(e\circ m)=\id_{T(pX)}$. Therefore
  $[\dot TY]^{pY}X\le\dot TX$.

\item 
  From the previous equality, we have
  \begin{displaymath}
    \phi_{Y,X}\circ\phi_{Z,Y}(S)=\phi_{Y,X}(q_{Y}([S]^{pZ}))= q_{X}([S]^{pZ})=\phi_{Z,X}(S).
    \tag*{\qedhere}
  \end{displaymath}

\end{enumerate}
\end{proof}
Following this lemma, we extend $\Cls(\mT,-)$ to a functor of type
$\Split(\EE)^{op}\arrow\Set$ by
\begin{displaymath}
  \Cls(\mT,X\lhd Y)=\phi_{Y,X}:\Cls(\mT,Y)\arrow\Cls(\mT,X).  
\end{displaymath}
This is indeed a functor thanks to Theorem \ref{th:adj} (for the
preservation of the identity) and Lemma \ref{lm:spcone}-2 (for the
preservation of the composition).

We establish a universal property of $\Lift(\mT)$
with respect to a restricted class of cones over $\Cls(\mT,-)$.
\begin{definition}
  Let $V$ be a class and $\{r_X:V\arrow\Cls(\mT,X)\}_{X\in\EE}$ be a
  cone from $V$ over $\Cls(\mT,-)$. We say that the cone $r$
  satisfies $\phi$-{\em inequality} if $\phi_{Y,X}\circ r_Y \ge r_X$
  holds for {\em any} $X,Y\in\EE$.
\end{definition}
From Lemma \ref{lm:spcone}-1, $ \{q_X:\Lift(\mT)\arrow
\Cls(\mT,X)\}_{X\in\EE} $ is a cone from $\Lift(\mT)$ over
$\Cls(\mT,-)$, and moreover it satisfies the $\phi$-inequality by
\eqref{eq:lax}.
\begin{theorem}\label{th:limit}
  For any class $V$ and cone $r$ from $V$ over $\Cls(\mT,-)$
  satisfying $\phi$-inequality, 
  there exists a unique function $m:V\arrow\Lift(\mT)$
  such that $r_X=q_X\circ m$ holds for any $X\in\EE$.
\end{theorem}
\begin{proof}
  Let $V$ be a class and $r$ be a cone from $V$ over
  $\Cls(\mT,-)$ satisfying $\phi$-inequality. We define the function
  $m:V\arrow\Lift(\mT)$ by
  \begin{displaymath}
    m(v)=\bigwedge_{X\in\EE}[r_X(v)]^{pX}.
  \end{displaymath}
  We show that $m$ satisfies $q_Y\circ m=r_Y$ for any $Y\in\EE$.
  \begin{displaymath}
    q_Y(m(v))=
    m(v)(Y)=
    \bigwedge_{X\in\EE}\phi_{X,Y}(r_X(v))=
    r_Y(v)\wedge
    \bigwedge_{X\in\EE,Y\not\lhd X}\phi_{X,Y}(r_X(v))=
    r_Y(v).
  \end{displaymath}
  If there is another function $m':V\arrow\Lift(\mT)$ such that $q_Y\circ m'= r_Y$ then
  from Theorem \ref{th:intersection}, we have
  \begin{displaymath}
    m'(v)=\bigwedge_{X\in\EE}[q_X(m'(v))]^{pX}=\bigwedge_{X\in\EE}[r_X(v)]^{pX}=m(v).
  \end{displaymath}
  Thus $m= m'$.
\end{proof}
In Theorem 29 of the conference version of this paper \cite{calco15},
we showed a different universal property about the cone $q$ from
$\Lift(\mT)$. There, we considered all cones (which may not satisfy
$\phi$-inequality), while we restricted $\Split(\EE)$ to be
directed. We later realised that 1) $\Split(\EE)$ is often not
directed due to the initial object in $\EE$ (this happens, for
instance, when $\EE=\Pred,\Top,\Pre,\Met$), and 2) the
$\phi$-inequality property makes the proof of the universal property
of $q$ work. We therefore changed the claim of \cite[Theorem
29]{calco15} to Theorem \ref{th:limit}.



\section{Conclusion and Future Work}
\label{sec:}

We introduced the codensity lifting of monads along the fibrations
that preserve the right Kan extensions giving codensity monads (this
codensity condition was relaxed later in Section
\ref{sec:without}). The codensity lifting allows us to lift various
monads on non-closed base categories, which was not possible
by its precursor, $\top\top$-lifting \cite{katsumatatt}. The categorical
dual of the codensity lifting is also given, which lifts
comonads along cofibrations.

\section*{Acknowledgments}

The derivation of Kantorovich metric by the
codensity lifting of Giry monad in Section \ref{sec:kantoro} was constructed after
Katsumata learned about
Ogawa's work on deriving Kantorovich metric for
finitely-supported subdistributions using observational algebra
\cite{ogawa}. Katsumata is grateful to him for the discussion
about Kantorovich metric and pseudometric spaces at CSCAT 2015. 

Katsumata and Sato were supported by JSPS KAKENHI grant no.\ 24700012.
Uustalu was supported by the ERDF funded Estonian ICT R\&D support
programme project ``Coinduction'' (3.2.1201.13-0029), the Estonian
Science Foundation grant no. 9475 and the Estonian Ministry of
Education and Research institutional research grant no.\ IUT33-13.


\bibliographystyle{alpha}
\bibliography{all}

\begin{thebibliography}{vBMOW05}

\bibitem[Aba00]{Abaditt}
M.~Abadi.
\newblock $\top\top$-closed relations and admissibility.
\newblock {\em Math. Struct.Comput. Sci.}, 10(3):313--320, 2000.

\bibitem[Bar70]{barrext}
M.~Barr.
\newblock Relational algebras.
\newblock In S.~{Mac Lane} et~al., editors, {\em Reports of the Midwest
  Category Seminar IV}, volume 137 of {\em Lecture Notes in Mathematics}, pages
  39–--55. Springer, Berlin, Heidelberg, 1970.

\bibitem[BBLM14]{BBLM2014}
G.~Bacci, G.~Bacci, K.~G. Larsen, and R.~Mardare.
\newblock Bisimulation on {M}arkov processes over arbitrary measurable spaces.
\newblock In F.~van Breugel, E.~Kashefi, C.~Palamidessi, and J.~Rutten,
  editors, {\em Horizons of the Mind: A Tribute to Prakash Panangaden}, volume
  8464 of {\em Lecture Notes in Computer Science}, pages 76--95. Springer,
  Cham, 2014.

\bibitem[Fil94]{filinskirepre}
A.~Filinski.
\newblock Representing monads.
\newblock In {\em Conf. Record of 21st ACM SIGPLAN-SIGACT Symp. on Principles
  of Programming Languages, POPL~'94 (Portland, OR, Jan. 1994)}, pages
  446--457. ACM, New York, 1994.

\bibitem[Fil96]{filinskiphd}
A.~Filinski.
\newblock {\em Controlling Effects}.
\newblock PhD thesis, Carnegie Mellon University, 1996.

\bibitem[Fil07]{filinskicompare}
A.~Filinski.
\newblock On the relations between monadic semantics.
\newblock {\em Theor. Comput. Sci.}, 375(1-3):41--75, 2007.

\bibitem[FS07]{filinskistrovring}
A.~Filinski and K.~St{\o}vring.
\newblock Inductive reeasoning about effectful data types.
\newblock In {\em Proc. of 12th ACM SIGPLAN Int. Conf. on Functional
  Programming, ICFP 2007 (Freiburg, Oct. 2007)}, pages 97--110. ACM, New York,
  2007.

\bibitem[Gir82]{gilly}
M.~Giry.
\newblock A categorical approach to probability theory.
\newblock In B.~Banaschewski, editor, {\em Categorical Aspects of Topology and
  Analysis}, volume 915 of {\em Lecture Notes in Mathematics}, pages 68--85.
  Springer, Berlin, Heidelberg, 1982.

\bibitem[Gir87]{girardlinear}
J.-Y. Girard.
\newblock Linear logic.
\newblock {\em Theor. Comp. Sci.}, 50:1--102, 1987.

\bibitem[Gir89]{goi}
J.-Y. Girard.
\newblock Geometry of interaction 1: Interpretation of system {F}.
\newblock In R.~Ferro, C.~Bonotto, S.~Valentini, and A.~Zanardo, editors, {\em
  Logic Colloquium~'88}, volume 127 of {\em Studies in Logic and the
  Foundations of Mathematics}, pages 221--260. North Holland, Amsterdam, 1989.

\bibitem[Her93]{hermidathesis}
C.~Hermida.
\newblock {\em Fibrations, Logical Predicates and Indeterminants}.
\newblock PhD thesis, University of Edinburgh, 1993.

\bibitem[HJ98]{hermidajacobs}
C.~Hermida and B.~Jacobs.
\newblock Structural induction and coinduction in a fibrational setting.
\newblock {\em Inf. Comput.}, 145(2):107--152, 1998.

\bibitem[HT00]{hesselinkthijs}
W.~H. Hesselink and A.~Thijs.
\newblock Fixpoint semantics and simulation.
\newblock {\em Theor. Comput. Sci.}, 238(1--2):275--311, 2000.

\bibitem[Jac99]{jacobscltt}
B.~Jacobs.
\newblock {\em Categorical Logic and Type Theory}, volume 141 of {\em Studies
  in Logic and the Foundations of Mathematics}.
\newblock Elsevier, Amsterdam, 1999.

\bibitem[JH03]{hughesjacobs}
B.~Jacobs and J.~Hughes.
\newblock Simulations in coalgebra.
\newblock {\em Electron. Notes Theor. Comput. Sci.}, 82(1):128--149, 2003.

\bibitem[Kan42]{kantorovich}
L.~Kantorovich.
\newblock On the translocation of masses (in {R}ussian).
\newblock {\em Doklady Akademii Nauk}, 37(7--8):227--229, 1942.
\newblock Translated in Management Science, 5(1):1--4, 1958.

\bibitem[Kat05]{katsumatatt}
S.~Katsumata.
\newblock A semantic formulation of $\top\top$-lifting and logical predicates
  for computational metalanguage.
\newblock In C.-H.~L. Ong, editor, {\em Proc. of 19th Int. Workshop on Computer
  Science Logic, CSL~2005 (Oxford, Aug. 2005)}, volume 3634 of {\em Lecture
  Notes in Computer Science}, pages 87--102. Springer, Berlin, Heidelberg,
  2005.

\bibitem[Kat13]{katsumatarelating}
S.~Katsumata.
\newblock Relating computational effects by $\top\top$-lifting.
\newblock {\em Inf. Comput.}, 222:228--246, 2013.

\bibitem[Kri09]{krivineclassical}
J.-L. Krivine.
\newblock Realizability in classical logic.
\newblock In P.-L. Curien, H.~Herbelin, J.-L. Krivine, and P.-A. Melli{\`e}s,
  editors, {\em Interactive Models of Computation and Program Behaviour},
  volume~27 of {\em Panoramas et Synth\`eses}, pages 197--229. Soci{\'e}t{\'e}
  Math{\'e}matique de France, Paris, 2009.

\bibitem[KS13]{DBLP:conf/fossacs/KatsumataS13}
S.~Katsumata and T.~Sato.
\newblock Preorders on monads and coalgebraic simulations.
\newblock In F.~Pfenning, editor, {\em Proc. of 16th Int. Conf. on Foundations
  of Software Science and Computation Structures, FOSSACS 2013 (Rome, March
  2013)}, volume 7794 of {\em Lecture Notes in Computer Science}, pages
  145--160. Springer, Berlin, Heidelberg, 2013.

\bibitem[KS15]{calco15}
S.~Katsumata and T.~Sato.
\newblock Codensity liftings of monads.
\newblock In L.~S. Moss and P.~Sobocinski, editors, {\em Proc. of 6th Conf. on
  Algebra and Coalgebra in Computer Science, CALCO 2015 (Nijmegen, June 2015)},
  volume~35 of {\em Leibniz Int. Proc. in Informatics}, pages 156--170.
  Dagstuhl Publishing, Saarbr\"ucken, Wadern, 2015.

\bibitem[Lev11]{similarq}
P.~B. Levy.
\newblock Similarity quotients as final coalgebras.
\newblock In M.~Hofmann, editor, {\em Proc. of 14th Int. Conf. on Foundations
  of Software Science and Computation Structures, FoSSaCS 2011
  (Saarbr{\"u}cken, March/Apr. 2011)}, volume 6604 of {\em Lecture Notes in
  Computer Science}, pages 27--41. Springer, Berlin, Heidelberg, 2011.

\bibitem[Lin05]{lindley}
S.~Lindley.
\newblock {\em Normalisation by Evaluation in the Compilation of Typed
  Functional Programming Languages}.
\newblock PhD thesis, University of Edinburgh, 2005.

\bibitem[LLN08]{lrmontyp}
J.-G. Larrecq, S.~Lasota, and D.~Nowak.
\newblock Logical relations for monadic types.
\newblock {\em Math. Struct. Comput. Science}, 18(6):1169--1217, 2008.

\bibitem[LS05]{lindleystark}
S.~Lindley and I.~Stark.
\newblock Reducibility and $\top\top$-lifting for computation types.
\newblock In P.~Urzyczyn, editor, {\em Proc. of 7th Int. Conf. on Typed Lambda
  Calculi and Applications, TLCA 2005 (Nara, Apr. 2005)}, volume 3461 of {\em
  Lecture Notes in Computer Science}, pages 262--277. Springer, Berlin,
  Heidelberg, 2005.

\bibitem[Mac98]{cwm2}
S.~MacLane.
\newblock {\em Categories for the Working Mathematician}, volume~5 of {\em
  Graduate Texts in Mathematics}.
\newblock Springer, New York, 2nd edition, 1998.

\bibitem[Mog91]{moggicomputational}
E.~Moggi.
\newblock Notions of computation and monads.
\newblock {\em Inf. Comput.}, 93(1):55--92, 1991.

\bibitem[MR92]{mareynolds}
Q.~Ma and J.~Reynolds.
\newblock Types, abstractions, and parametric polymorphism, part 2.
\newblock In S.~D. Brookes, M.~G. Main, A.~Melton, M.~W. Mislove, and D.~A.
  Schmidt, editors, {\em Proc. of 7th Int. Conf. on Mathematical Foundations of
  Programming Semantics, MFPS '91 (Pittsburgh, PA, March 1991)}, volume 598 of
  {\em Lecture Notes in Computer Science}, pages 1--40. Springer, Berlin,
  Heidelberg, 1992.

\bibitem[MS93]{mitsce}
J.~Mitchell and A.~Scedrov.
\newblock Notes on sconing and relators.
\newblock In E.~B{\"o}rger, H.~Kleine B{\"u}ning, S.~Martini, and M.~M.
  Richter, editors, {\em Selected Papers from 6th Workshop on Computer Science
  Logic, CSL~'92 (San Miniato, Sept./Oct. 1992)}, volume 702 of {\em Lecture
  Notes in Computer Science}, pages 352--378. Springer, Berlin, Heidelberg,
  1993.

\bibitem[MV15]{Marti2015880}
J.~Marti and Y.~Venema.
\newblock Lax extensions of coalgebra functors and their logic.
\newblock {\em J. Comput. Syst. Sci.}, 81(5):880--900, 2015.

\bibitem[MZ15]{mellieszeilberger}
P.-A. Melli\`{e}s and N.~Zeilberger.
\newblock Functors are type refinement systems.
\newblock In {\em Proc. of 42nd Ann. ACM SIGPLAN-SIGACT Symp. on Principles of
  Programming Languages, POPL '15 (Mumbai, Jan. 2015)}, pages 3--16. ACM, New
  York, 2015.

\bibitem[{Nad}78]{nadler}
S.~B. {Nadler, Jr.}
\newblock {\em Hyperspaces of Sets: A Text with Research Questions}, volume~49
  of {\em Monographs and Textbooks on Pure and Applied Mathematics}.
\newblock Marcel Dekker, 1978.

\bibitem[Oga15]{ogawa}
H.~Ogawa.
\newblock Quotient and {K}antorovich metric via observational-algebra in
  {L}awvere theory.
\newblock Oral presentation at CSCAT 2015, Kagoshima University, Japan, 14
  March 2015.

\bibitem[Pit00]{parpolyopeq}
A.~Pitts.
\newblock Parametric polymorphism and operational equivalence.
\newblock {\em Math. Struct. Comput. Sci.}, 10(3):321--359, 2000.

\bibitem[Plo80]{plotkindef}
G.~D. Plotkin.
\newblock Lambda-definability in the full type hierarchy.
\newblock In J.~P. Seldin and J.~R. Hindley, editors, {\em To H. B. Curry:
  Essays on Combinatory Logic, Lambda Calculus and Formalism}, pages 367--373.
  Academic Press, San Diego, CA, 1980.

\bibitem[PP03]{DBLP:journals/acs/PlotkinP03}
G.~D. Plotkin and J.~Power.
\newblock Algebraic operations and generic effects.
\newblock {\em Appl. Categ. Struct.}, 11(1):69--94, 2003.

\bibitem[SKDH18]{DBLP:conf/cmcs/SprungerKDH18}
D.~Sprunger, S.~Katsumata, J.~Dubut, and I.~Hasuo.
\newblock Fibrational bisimulations and quantitative reasoning.
\newblock In Corina C{\^{\i}}rstea, editor, {\em Coalgebraic Methods in
  Computer Science - 14th {IFIP} {WG} 1.3 International Workshop, {CMCS} 2018,
  Colocated with {ETAPS} 2018, Thessaloniki, Greece, April 14-15, 2018, Revised
  Selected Papers}, volume 11202 of {\em Lecture Notes in Computer Science},
  pages 190--213. Springer, 2018.

\bibitem[UV08]{uustaluvene}
T.~Uustalu and V.~Vene.
\newblock Comonadic notions of computation.
\newblock {\em Electron. Notes Theor. Comput. Sci.}, 203(5):263--284, 2008.

\bibitem[vBMOW05]{DBLP:journals/tcs/BreugelMOW05}
F.~van Breugel, M.~W. Mislove, J.~Ouaknine, and J.~Worrell.
\newblock Domain theory, testing and simulation for labelled {M}arkov
  processes.
\newblock {\em Theor. Comput. Sci.}, 333(1-2):171--197, 2005.

\bibitem[WV04]{wandvaillancourt}
M.~Wand and D.~Vaillancourt.
\newblock Relating models of backtracking.
\newblock In {\em Proc.\ of 9th ACM SIGPLAN Int.\ Conf.\ on Functional
  Programming, ICFP~'04 (Snowbird, UT, Sept. 2004)}, pages 54--65. ACM, New
  York, 2004.

\end{thebibliography}
\end{document}